\documentclass[a4paper,12pt,english]{article}

\usepackage{amsfonts,bm,amssymb,amsthm,euscript,array,babel,cite}
\usepackage{mathtools}
\usepackage{tikz-cd}
\usepackage{amsmath,comment}
\usepackage{tensor,accents}
\usepackage{hhline}
\usepackage[amsmath]{empheq}
\usepackage{hyperref}
\usepackage[utf8]{inputenc}
\usepackage[T1]{fontenc}
\usepackage{cancel}
\usepackage{multirow}
\usepackage{mathrsfs}
\usepackage{pdfsync}
\usepackage{multirow}
\usepackage{tikz}
\usepackage[most]{tcolorbox}

\usepackage{framed}

\newtcbox{\mymath}[1][]{%
    nobeforeafter, math upper, tcbox raise base,
    enhanced, colframe=blue!30!black,
    colback=blue!30, boxrule=1pt,
    #1}

\newsavebox{\mysaveboxM}
\newsavebox{\mysaveboxT}

\makeatletter

\newcommand{\dd}{\mathrm{d}}
\newcommand{\DD}{\mathrm{D}}

\newcommand{\w}{\wedge}

\newcommand{\be}{\begin{equation}}
\newcommand{\ee}{\end{equation}}

\def\nn{\nonumber}
\def \bea{\begin{eqnarray}} 
\def\eea{\end{eqnarray}}
\newcommand{\mf}{\mathfrak}
\def\mc{\mathcal}

\def\bi{\begin{itemize}} 
\def\ei{\end{itemize}}

\def\E{\textit{\tiny{E}}} 

\def\Eh{\widehat{\textit{\tiny{E}}}} 
\newcommand{\sbullet}{%
  \hbox{\fontfamily{lmr}\fontsize{.8\dimexpr(\f@size pt)}{0}\selectfont\textbullet}}
\DeclareRobustCommand{\mathbullet}{\accentset{\sbullet}}

\makeatother

\newtheorem{prop}[equation]{Proposition}
\newtheorem{defn}[equation]{Definition}
\newtheorem{rmk}[equation]{Remark}

\newtheorem{exa}[equation]{Example}

\def\a{\alpha} \def\b{\beta}  \def\G{\Gamma} \def\d{\delta} 
\def\e{\epsilon} 
   \def\k{\kappa}
\def\l{\lambda}  \def\m{\mu}
\def\n{\nu}    \def\r{\rho}
\def\s{\sigma} \def\S{\Sigma}

\def\H{\textit{\tiny{H}}}

\def\one{\mbox{1 \kern-.59em {\rm l}}}

\numberwithin{equation}{section}

\begin{document}

\makeatother
\parindent=0cm
\renewcommand{\title}[1]{\vspace{10mm}\noindent{\Large{\bf #1}}\vspace{8mm}} \newcommand{\authors}[1]{\noindent{\large #1}\vspace{5mm}} \newcommand{\address}[1]{{\itshape #1\vspace{2mm}}}

\begin{titlepage}
	
	\begin{flushright} 
RBI-ThPhys-2024-06
\end{flushright}
	
	\begin{center}
		
		
		\title{ {\Large Brane mechanics and gapped Lie n-algebroids }}
		
		\vskip 3mm
		
		  \authors{ Athanasios Chatzistavrakidis$^{\,\a}$, Toni Kod\v{z}oman$^{\,\a}$, Zoran \v{S}koda$^{\,\b}$
		\large
		    }
		 
		 \vskip 3mm
		 
		  \address{ $^{\a}$\,Division of Theoretical Physics, Ruđer Bo\v skovi\'c Institute \\ Bijeni\v cka 54, 10000 Zagreb, Croatia 
    
    \ 
    
    $^{\b}$ Faculty of Teacher's Education, University of Zadar \\  F. Tudjmana 24, 23000 Zadar, Croatia}
		
		\vskip 2cm
		
		\begin{abstract}
			\noindent
			We draw a parallel between the BV/BRST formalism for higher-dimensional ($\ge 2$) Hamiltonian mechanics and higher notions of torsion and basic curvature tensors for generalized connections in specific Lie $n$-algebroids based on homotopy Poisson structures. The gauge systems we consider include Poisson sigma models in any dimension and ``generalised R-flux'' deformations thereof, such as models with an $(n+2)$-form-twisted R-Poisson target space. Their BV/BRST action includes interaction terms among the fields, ghosts and antifields whose coefficients acquire a  geometric meaning by considering twisted Koszul multibrackets that endow the target space with a structure that we call a gapped almost Lie $n$-algebroid. Studying  covariant derivatives along $n$-forms, we define suitable polytorsion and basic polycurvature tensors and identify them with the interaction coefficients in the gauge theory, thus relating models for topological $n$-branes to differential geometry on Lie $n$-algebroids. 
		\end{abstract}
		
	\end{center}
	
	\vskip 2cm
	
\end{titlepage}

\tableofcontents 

\section{Introduction}
\label{sec0} 

In recent years there is a rekindled  interest in higher structures in classical and quantum field theory. A major role in this has been played by the BRST formalism for quantization of gauge systems, see e.g. \cite{HTbook,Barnich:2000zw} for context and  literature, which has often been a guide for uncovering novel mathematical structures behind field theory. These include (strong) homotopy algebras \cite{Zwiebach:1992ie,Lada:1992wc} and their categorical extensions to algebroids and the (super)geometry of differential graded manifolds. 

Most approaches, especially in the physics literature, build primarily on the algebraic aspects of such structures. Specifically, the usual objects of interest are higher brackets, see for example the general discussion of the relation between perturbative gauge field theories and $L_{\infty}$ algebras in \cite{Hohm:2017pnh}. Less light has been shed on differential geometric aspects such as connections and their torsion and curvature tensors. An exception is Courant algebroids whose natural appearance in string theory through generalised geometry \cite{Hitchin,Gualtieri} sparked a study of generalised connections and tensors on them, see e.g. \cite{Gualtieri:2007bq}. From a physics perspective such a geometrical viewpoint is natural, not only for aesthetic or conceptual reasons but also for the construction of (super)gravity theories, as for example in Refs. \cite{Coimbra:2011nw,Hohm:2010pp,Hohm:2012mf,Jurco:2016emw,Severa:2018pag,Boffo:2019zus,Jurco:2022nlr,Aschieri:2019qku}. 

 Within the broader context of establishing relations between higher geometric structures and field theory, in this paper we aim at providing a complete geometric explanation of the structural data that appear in the BV/BRST formulation of a certain class of topological sigma models in any dimension, the so-called twisted R-Poisson sigma models introduced in \cite{Thanasis R-Poisson} and further generalized in \cite{Ikeda:2021rir}. These are explicit and  tractable models of \v{S}evera's general theme of higher-dimensional Hamiltonian mechanics, which in turn present a way to model symplectic Lie $n$-algebroids (termed ``$\S_{n}$-manifolds'') as variational problems \cite{Pavol some title}. The theme follows the logical route ``$n=0$ is for symplectic, $n=1$ is for Poisson (particle mechanics), $n=2$ is for Courant (string mechanics), $n=$ many is for $\Sigma_{n}$-manifolds (brane mechanics)''. 

Twisted R-Poisson sigma models were motivated from two sides. On one hand from world volume studies of closed string backgrounds with ``R-flux''  \cite{Halmagyi:2008dr,Mylonas:2012pg,Chatzistavrakidis:2015vka,Bessho:2015tkk}; in that case R is a trivector and the structure is actually untwisted. On the other hand, from the viewpoint of twisting geometric structures on target spaces of sigma models. A prototype in this spirit is the Wess-Zumino-Poisson sigma model in two dimensions \cite{Klimcik:2001vg}, in which case the would-be Poisson structure on the target space is twisted by a 3-form   and Poisson-ness is obstructed \cite{Severa Weinstein}. 

Twisted R-Poisson structures and their associated higher dimensional WZW-like sigma models are the answer to various questions about extensions of these themes. They provide generalizations of R-flux models for higher-dimensional branes with an $(n+1)$-vector deformation, as well as a manifestly target space covariant formulation of the three- and higher-dimensional cases. They also provide an answer to the question \emph{``what is a Poisson sigma model in dimensions greater than 2?''}, which we will clarify in the present paper. Twisted R-Poisson sigma models go beyond this realm though, including Wess-Zumino terms that lead to a certain twisting. Thus in 3D they provide concrete realizations of twisted Courant sigma models with 4-form Wess-Zumino term \cite{Hansen:2009zd}, associated to pre-Courant algebroids \cite{Vaisman:2004msa}, and more generally they yield concrete higher generalizations of these structures. Viewed differently, they are possibly the simplest gauge theories that feature simultaneously all the main characteristics that the BV formalism was constructed for, namely an open gauge algebra, field-dependent structure functions and reducibilities. What is more, they exhibit the unorthodox feature of a non-linearly open gauge algebra, since the latter closes on products of field equations.

The untwisted version of these theories is a slice in the space of AKSZ sigma models in various dimensions \cite{Alexandrov:1995kv}, with the BV master action of the Poisson sigma model being the 2D case \cite{Cattaneo:1999fm,Cattaneo:2001ys}. The twisted version refers to Wess-Zumino-AKSZ theories \cite{Ikeda:2013wh}, which are less studied. The new, Wess-Zumino-induced, contributions to the solution of the classical master equation of these models appear in a derivative expansion of the AKSZ action in terms of superfields and therefore they are somewhat hidden. In other words, they appear in the coefficients of the interaction terms in the master action, which usually involve several ghosts and antifields and they are only visible in the expanded form of the action. A  geometrical treatment of such contributions in the 2D case \cite{Ikeda:2019czt} revealed that they are accounted for with the aid of a connection on the relevant cotangent Lie algebroid, in which case the interaction coefficients become torsion and basic curvature tensors. Retrospectively, as explained in \cite{Chatzistavrakidis:2023otk}, this is in line with expectations since the same happens in the topological A and B models{\footnote{For the relation of the Poisson sigma model and the A or B model, see \cite{Alexandrov:1995kv,Bonechi:2016wqz}.}} \cite{Witten:1988xj}. Note that auxiliary connections were introduced to restore target space covariance and discuss global issues in topological field theories in \cite{Baulieu:2001fi,Cattaneo:2000hz,Bojowald:2004wu}. Analogous results in the general 3D case of 4-form twisted Courant sigma models appear in \cite{Chatzistavrakidis:2023lwo}.   

In this work we take this correspondence between field theory and geometry one step further by demonstrating the relation of the interaction coefficients in the master action of twisted R-Poisson sigma models to new tensorial objects for generalised connections on a specific class of Lie $n$-algebroids. In particular, we first identify the covariant form of these coefficients. Then we introduce higher Koszul brackets on the cotangent bundle of the target space in the same way as in \cite{Voronov2} and define gapped Lie $n$-algebroids, which are generated by these multibrackets. Note that the multibrackets are not $C^{\infty}$-linear but instead they satisfy Leibniz rules with respect to higher anchor maps. Gapped Lie $n$-algebroids are to R-Poisson manifolds (of order $n+1$) what Lie algebroids are to Poisson manifolds. Crucially, twisted structures depart from this definition and unlike twisted Poisson manifolds that still give rise to Lie algebroids, twisted R-Poisson manifolds are gapped \emph{almost} Lie $n$-algebroids.

Once the basic mathematical structure underlying the gauge theory is identified, we turn our attention to connections and tensors. We define covariant differentiation along $n$-forms and use it to introduce two new tensorial objects. One of them is similar to the torsion tensor, albeit it is a vector-valued $(n+1)$-form instead. We choose to call it the polytorsion tensor, since for $n=1$ it agrees with the usual $E$-torsion tensor of a Lie algebroid $E$-connection. The other is a generalization of the so-called basic curvature tensor \cite{Blaom,Crainic}, which measures the compatibility of a connection on a Lie algebroid with the binary Lie bracket on its sections. Here we define a basic polycurvature tensor that measures the compatibility of a connection with the higher $(n+1)$-ary Koszul bracket on the gapped Lie $n$-algebroid.  Computing the explicit expressions of the polytorsion and basic polycurvature tensors, we indeed find that they are identical to the interaction coefficients in the twisted R-Poisson sigma models and consequently they appear in their BRST differential. 

The manuscript is organised in the following way. In Section \ref{sec2} we present the topological sigma models that are the main focus of our analysis, with emphasis on aspects that have not been highlighted before. In particular, Section \ref{sec21} contains the answer to the question: \emph{What is a Poisson sigma model in more than two dimensions?} In Section \ref{sec22} we proceed with a generalised R-flux deformation that naturally leads to twisted R-Poisson sigma models and in Section \ref{sec23} we present their formulation as AKSZ theories, ending up with the interaction coefficients that we would like to account for geometrically. In Section \ref{sec3} we introduce higher Koszul brackets on the cotangent bundle of the target space and analyse the possibilities for twisting these brackets. Section \ref{sec4} contains the answer to the question: \emph{Which algebroid is for a (twisted) R-Poisson manifold what a Lie algebroid is for a (twisted) Poisson manifold?} This leads us to gapped (almost) Lie $n$-algebroids, for which we also discuss their description as differential graded manifolds. Section \ref{sec5} discusses generalized connections for such Lie $n$-algebroids and provides a definition of the polytorsion and its relation to the BRST coefficients. In a similar spirit, Section \ref{sec6} discusses the basic polycurvature tensor. Finally, we present our conclusions in Section \ref{sec7}. 

\section{Brane mechanics on a twisted R-Poisson manifold}
\label{sec2}

\subsection{Poisson sigma models in higher dimensions \& WZ terms} 
\label{sec21} 

We begin with a class of topological field theories whose target space is a Poisson manifold. Usually this refers to the Poisson sigma model in two dimensions \cite{Schaller:1994es,Ikeda:1993fh}. However, one may model Poisson manifolds as variational problems in any dimension as can be invoked by the more general construction in \cite{Thanasis R-Poisson}. Hence we use this more general perspective and think of a Poisson sigma model as one which (i) contains a scalar sector that is modelled on maps from a source manifold of any dimension to a Poisson manifold $(M,\Pi)$, and (ii) exhibits gauge symmetries such that the gauge invariance of the action functional is based solely on the condition of $\Pi$ being a Poisson structure, namely on the Jacobi identity of the associated Poisson bracket. This section discusses these models and the possibilities regarding their extension by a Wess-Zumino term; in the present case this is very restrictive, at least with the minimal data we  assume here, with the 2D model being possible to extend using the twisted Poisson structure of \cite{Severa Weinstein} as shown in \cite{Klimcik:2001vg}, whereas the higher dimensional models yield an untwisted extension with a closed $(n+2)$-form. However, this analysis will set up the problem for the more general models to be discussed in Section \ref{sec22} and their geometric interpretation in terms of the structures introduced in Sections \ref{sec3}-\ref{sec6}.  

To formulate the field theory we need some minimal data which we briefly introduce here; their place in a more general context will be discussed in the following sections. Consider a Poisson manifold $(M,\Pi)$; by definition the 2-vector $\Pi$ is a Poisson structure satisfying 
\be [\Pi,\Pi]_{\text{SN}}=0\,, \ee
in terms of the Schouten-Nijenhuis bracket for multivector fields, whose precise definition in our conventions is given later in Section \ref{sec3}. Presently it suffices to recall that in local coordinates this Poisson condition takes the familiar form
\be 
\Pi^{\k[\m}\partial_\k\Pi^{\n\rho]}=0\,,
\ee 
where the brackets denote antisymmetrization with weight 1 in all indices within them, here 3, and the components of the Poisson structure in a coordinate basis are $\Pi^{\m\n}=-\Pi^{\n\m}$. 
A Poisson manifold gives rise to a Lie algebroid on its cotangent bundle, which will also be described in Section \ref{sec3}. Alternatively, for the present discussion it is more useful to think in terms of the corresponding dg manifold, namely in terms of a homological vector field \cite{Vaintrob}. The cotangent Lie algebroid can be represented as such a vector field on the parity reversed cotangent bundle whose fiber degree is shifted by 1, denoted $T^{\ast}[1]M$. Denoting moreover by $x^{\m}$ a local coordinate system on $M$ and by $a_{\m}$ the degree 1 fiber coordinates, the vector field 
\be \label{Q Poisson}
Q^{\ast}=\Pi^{\m\n}(x)a_\m\frac{\partial}{\partial x^\n}-\frac 12 \partial_\rho\Pi^{\m\n}(x)a_\m a_\n \frac{\partial}{\partial a_\rho}
\ee  
is homological, namely the graded Lie bracket with itself vanishes, if and only if $(M,\Pi)$ is a Poisson manifold: 
\be 
(Q^{\ast})^2=\frac 12 [Q^{\ast},Q^{\ast}]=0 \qquad \Longleftrightarrow \qquad [\Pi,\Pi]_{\text{SN}}=0\,.
\ee 
The asterisk is just a notation that reminds us that we work on the cotangent bundle. Denoting $\mc{M}=T^{\ast}[1]M$, the pair $(\mc M,Q^{\ast})$ is a dg manifold, also called an NQ1 manifold in this context, since all coordinates have non negative degree---we skip the N in the notation from now on. In fact it is even a QP1 manifold, P referring to the Poisson structure on $\mc M$ that naturally exists since the dg manifold is a cotangent bundle. Poisson manifolds and QP1 structures are in one to one correspondence \cite{Roytenberg:2002nu}. 

This is, however, not the only way that we can encode a Poisson structure in a homological vector field. Rather it is a special ``$n=1$''-type case of a more general statement that is based on lifting structures to higher order bundles and goes as follows. Consider instead the graded manifold $\mc M_n=T^\ast[n]T^\ast[1]M$ with the base $M$ a Poisson manifold as before, without any additional structures defined on it; cf. \cite{Voronov2}. The graded manifold $\mc M_n$ can be endowed with a local coordinate system $\{x^\a\}=\{x^{\m},a_{\m},y^{\m},z_{\m}\}$ of degrees $0, 1, n-1, n$ respectively. It is a shifted second order cotangent bundle. Then the vector field 
\be
Q^{\ast}_{n}=Q^{\ast}+\big((-1)^{n}\Pi^{\m\n}z_\m-\partial_\l\Pi^{\n\m}a_\m y^{\l}\big)\frac{\partial}{\partial{y^{\n}}}  -  \left(\partial_\k\Pi^{\m\n}a_\m z_\n+\frac {(-1)^{n}}2 \partial_\k\partial_\l\Pi^{\m\n}y^{\l}a_\m a_\n\right)\frac{\partial}{\partial{z_{\k}}}\label{Qn}
\ee
is homological and therefore $(\mc M_n, Q^{\ast}_n)$ is a dg manifold if and only if $(M,\Pi)$ is a Poisson manifold, as assumed.  
Note that this is a distinct yet not unrelated statement from the fact that  $n=1$ corresponds to Poisson manifolds and $n=2$ to Courant algebroids \cite{Roytenberg:2002nu}. Here we explained how to model a Poisson structure for any $n$. Nevertheless, the model for $n=2$ is a special case of the homological vector field of a Courant algebroid with the anchor given by the map corresponding to the Poisson structure. This Courant algebroid is a Lie bialgebroid formed by summing a totally intransitive Lie algebroid on the tangent bundle $TM$ and the cotangent Lie algebroid on the Poisson manifold.  Lie bialgebroids give rise to Courant algebroids in a straightforward manner \cite{Liu:1995lsa}. Similarly, for general $n$ this can also be seen to correspond to analogous structures for higher Courant algebroids as discussed in Ref. \cite{Ikeda:2021rir}.

Whereas the homological vector field in \eqref{Qn} is more complicated than the usual one in \eqref{Q Poisson}, it serves the purpose that it gives rise to topological field theories in any dimension, unlike the one in \eqref{Q Poisson} that is suitable only for 2-dimensional theories. Note moreover the asymmetry in the treatment; taking $n=1$ it turns out that $Q^{\ast}_1$ is not identical to $Q^{\ast}$. To avoid such confusion in the following we restrict to the case $n>1$ which is relevant in dimensions greater than 2. For 2 dimensions we reserve the standard homological vector field of Eq. \eqref{Q Poisson}.{\footnote{Note, however, that 2D models with target space $T^{\ast}[1]T^{\ast}[1]M$ are interesting in their own right as ``doubled sigma models'' in the sense that they comprise two set of real scalar fields \cite{Thanasis R-Poisson}. In a different direction they can also serve as the structure that leads to supersymmetric Poisson sigma models when the additional fields are assigned different parity \cite{Arias:2015wha,Arias:2016agc}.}}

Then the Poisson sigma model in $n+1$ dimensions is given by the action functional 
\be 
S^{(n+1)}_{\text{\,PSM}}=\int\left(Z_{\m}\w\dd X^{\m}-A_{\m}\w\dd Y^{\m}+\Pi^{\m\n}Z_{\m}\w A_{\n}-\frac 12 \partial_{\rho}\Pi^{\m\n}\, Y^{\rho}\w A_{\m}\w A_{\n}\right)\,.
\ee 
The fields are an $n$-form $Z_{\m}$, an $(n-1)$-form $Y^{\m}$, a 1-form $A_{\m}$ and the $\text{dim}\,M$ real scalars $X^{\m}$. We note that for $n=1$ the choice $A_{\m}=\frac 12 Z_{\m}$ and $Y^{\m}=0$ gives the usual 2-dimensional Poisson sigma model. Rather than discussing further these models, we are going to proceed with a more general class that contains them as special cases. We only note here that the field equations of the higher dimensional model include 
\bea 
F^{\m}&:=&\dd X^{\m}+\Pi^{\m\n}A_{\n}=0\,, \\[4pt] 
G_{\m}&:=& \dd A_{\m}+\frac 12 \partial_{\m}\Pi^{\n\rho}A_{\n}\w A_{\rho}=0\,,
\eea  
which are the same as the ones for the 2-dimensional model. There are of course two more field equations in this higher dimensional case. 

Before closing this section, let us ask what happens when the Poisson manifold $M$ is also equipped with a closed $(n+2)$-form $H\in \Omega_{\text{cl}}^{n+2}(M)$. Then we can extend the action functional by a Wess-Zumino term as 
\be \label{WZPSM n+1}
S^{(n+1)}_{\text{WZ-PSM}}=S^{(n+1)}_{\text{PSM}}+\int_{\S_{n+2}}X^{\ast}H\,,
\ee 
with the $(n+1)$-dimensional theory living on the boundary of $\S_{n+2}$ such that it does not depend on the higher brane invoked to support the Wess-Zumino term \cite{Witten:1983ar}. The gauge symmetries $\d$ of the extended model are the same as the non extended ones, say $\d_0$, save the gauge transformation of the top form field $Z_{\m}$, which changes to 
\be  
\d Z_{\m}=\d_0 Z_{\m}-\frac 1{(n+1)!}\Pi^{\k\l}H_{\m\l\n_1\dots\n_n}\Omega^{\n_1\dots\n_n}\e_{\k}\,
\ee 
where $\e_{\k}$ is the scalar gauge parameter of the model and the world volume $n$-form $\Omega$ is given by the following formula:
\be 
\Omega^{\n_1\dots\n_n}=\sum_{r=1}^{n+1}(-1)^{r}\prod_{s=1}^{r-1}\dd X^{\n_s}\prod_{t=r}^{n}\Pi^{\n_t\m_t}A_{\m_t}\,.
\ee 
This guarantees that the extended action functional is gauge invariant. We call this the higher dimensional Wess-Zumino-Poisson sigma model. It should be mentioned that only in 2 dimensions this leads to a twist of the Poisson structure. For dimensions 3 and higher the Poisson structure  cannot be twisted by $H$ and we end up simply with a Poisson structure and a closed $(n+2)$-form. A more interesting situation arises when we allow for an additional deformation given by a 3-vector as we discuss immediately below.

\subsection{Twisted R-Poisson brane mechanics}\label{sec22}

The models obtained through the simple lifting procedure described above can be deformed further. This is obvious from the fact that QP$n$ manifolds correspond to more general structures than Poisson manifolds. For $n=2$ they are Courant algebroids \cite{Roytenberg:2002nu} and for arbitrary $n$ they are symplectic Lie $n$-algebroids or ``$\S_n$-manifolds'' in the terminology of \cite{Pavol some title}. They give rise to higher dimensional Hamiltonian mechanics via the AKSZ construction as described in \cite{Pavol some title}. What we describe below, including the higher dimensional Poisson sigma models above, is a concrete realization of a particular class of such models in any dimension and of their extension by a Wess-Zumino term. We use the terminology ``brane mechanics'' to refer to all cases, twisted or not, for brevity. 

Let us clarify one important point that regards the Wess-Zumino terms that lead to twisted geometric structures. We mentioned already that 3-form twisted Poisson structures also give rise to a sigma model, even though they are not QP1 manifolds. They are, however, Q1 manifolds, namely Lie algebroids. Similarly for $n=2$ there exist sigma models with target space a Q2 manifold. These are 4-form twisted Courant sigma models \cite{Hansen:2009zd}, based on a relaxed structure arising from Courant algebroids by allowing violation of the Jacobi identity for the Dorfman bracket \cite{Vaisman:2004msa}. We will not discuss this further but rather we extend this logic to arbitrary dimensions and consider a special class of Q$n$ manifolds whose classical geometrical description is given in the rest of this paper. 

The (topological) field theory we consider is a deformation of the action functional \eqref{WZPSM n+1} of the $(n+1)$-dimensional Wess-Zumino Poisson sigma model. It features an additional $(n+1)$-vector $R=(R^{\m_1\dots\m_{n+1}})$ and its action functional is 
\be  
S^{(n+1)}_{\text{WZ-RPSM}}=S^{(n+1)}_{\text{WZ-PSM}}+\int \frac 1{(n+1)!}R^{\m_1\dots \m_{n+1}}A_{\m_1}\w \dots \w A_{\m_{n+1}}\,.
\ee 
This works in any dimension and the target space has the structure of a twisted R-Poisson manifold of order $n+1$ \cite{Thanasis R-Poisson}. This means that $\Pi$ is a (genuine) Poisson structure, $H$ is a closed $(n+2)$ form and $R$ is such that the following ($H$-twisted) condition holds: 
\be 
[\Pi,R]_{\text{SN}}=(-1)^{n+1}\langle\Pi^{\otimes(n+2)},H\rangle\,,
\ee
where the notation $\Pi^{\otimes(n+2)}$ means the $(n+2)$-fold tensor product of the 2-vector and the $n+2$ contractions are on the first index of each $\Pi$. 
We give more details and context for this structure in the next section. Much like the ordinary Poisson structure, the conditions for a twisted R-Poisson manifold can be encoded in a homological vector field that fully captures the classical gauge symmetries of the above model. It is given as 
\bea 
Q_H^{{RP_n}}=Q^{\ast}_n+\frac 1{n!}R^{\m\n_1\dots \n_{n}}a_{\n_1}\dots a_{\n_{n}}\frac{\partial}{\partial y^{\m}} 
  +\frac{(-1)^{n}}{(n+1)!}f_{\m}{}^{\n_1\dots\n_{n+1}}a_{\n_1}\dots a_{\n_{n+1}}\frac{\partial}{\partial z_{\m}}\,, \label{Q RPnH}
  \eea 
where we have defined  
\be 
  f_{\m}{}^{\n_1\dots\n_{n+1}}=\partial_{\m}R^{\n_1\dots \n_{n+1}}-H_{\m}{}^{\n_1\dots \n_{n+1}}\,,
  \ee 
and the indices in the $(n+2)$-form are raised with the 2-vector according to the convention
\be 
H^{\m_1\dots\m_{n+1}}{}_{\m}=\Pi^{\m_1\n_1}\dots\Pi^{\m_{n+1}\n_{n+1}}H_{\n_1\dots\n_{n+1}\m}\,.
 \ee 
The classical gauge symmetries of the twisted R-Poisson sigma model described above follow directly from the components of the homological vector field, as in any gauge theory, see e.g. \cite{Grutzmann:2014hkn}. Denoting the fields of the theory as $(\phi^{\a})=(X^{\m}, A_{\m}, Y^{\m}, Z_{\m})$, the gauge transformations for all fields read 
\be 
\d \phi^{\a}=\dd\e^{\a}+\e^{\b}\partial_{\b}Q^{\a}+\d_{\text{triv}}\phi^{\a}\,,
\ee 
where the gauge parameters are collectively denoted as $\e^{\a}$, $Q^{\a}$ are the components of the homological vector field $Q_{H}^{RP_n}$ and $\d_{\text{triv}}$ vanishes for all fields apart from $Z_{\m}$ for which 
\be 
\d_{\text{triv}}Z_{\m}=\frac 1{(n+1)!}\Pi^{\k\l}H_{\m\l\n_1\dots \n_n}\sum_{r=1}^{n}(-1)^{r}\binom{n+1}{r+1}\prod_{s=1}^{r} F^{\n_{s}}\prod_{t=r+1}^{n}\Pi^{\n_t\r_t}A_{\r_t}\epsilon_\k\,.
\ee 
We observe that it is given by an expression that contains a sum over products of the field equation of the field $Z_{\m}$. This is expected in view of the fact that the models we introduced have open gauge algebra and moreover this gauge algebra contains products of field equations. This property was called non-linear openness in \cite{CIS}. The 3D twisted R-Poisson sigma model seems to be the simplest gauge theory that exhibits this feature. 

\subsection{Interaction coefficients, Hamiltonian lift and BRST}\label{sec23}

The untwisted R-Poisson sigma models are a very simple and explicit case of AKSZ sigma models in arbitrary dimensions. As such they serve as an illustration of the power of the AKSZ formalism, since we can immediately write down their master action and prepare them for quantization.{\footnote{The complete quantization in higher than 2 dimensions is of course a hard problem.}} All we have to do is to first identify the extended sector of ghosts, ghosts of ghosts and antifields of the theory, define 4 superfields on the shifted tangent bundle $T[1]\S_{n+1}$ of the world volume and upgrade the classical action to a master action by writing it in terms of superfields. 

To give some details about what we just described in words, note that the full field content in the field/antifield formulation of the models is compactly written as 
\bea 
&&\text{classical fields} \quad X, A, Y, Z \quad \text{with} \quad (0,0), (0,1), (0,n-1), (0,n) \nn\\[4pt] 
&&\text{ghosts} \quad \epsilon, \chi, \psi \quad \text{with} \quad (1,0), (1,n-2), (1,n-1) \nn\\[4pt]
&& \text{ghosts of ghosts} \quad \chi_{(r)}, \psi_{(s)} \quad\text{with}\quad (r+1,n-r-2), (s+1,n-s-2)\,, \nn
\eea 
where we skipped the obvious index structure to avoid unnecessary clutter, the two towers of ghosts of ghosts are indexed as $r=1,\dots,n-2$ and $s=1,\dots,n-1$, and the parentheses contain the ghost degree and the form degree of each field respectively. Similarly for the antifields of all the fields above,
\bea 
&&X^{\dagger}, A^{\dagger}, Y^{\dagger}, Z^{\dagger} \quad \text{with} \quad (-1,n+1), (-1,n), (-1,2), (-1,1) \nn\\[4pt] 
&&\e^{\dagger}, \chi^{\dagger}, \psi^{\dagger} \quad \text{with} \quad (-2,n+1), (-2,3), (-2,2) \nn\\[4pt] 
&& \chi^{\dagger}_{(r)}, 
\psi^{\dagger}_{(s)} \quad \text{with} \quad (-r-2,r+3), (-s-2,s+2)\,. \nn
\eea
Taking the total degree to be the sum of the ghost and form degrees, we observe that all the $4(n+2)$ fields and antifields above organise themselves in the following four superfields, each with $n+2$ terms presented in order of ascending form degree on $\S_{n+1}$, namely from the scalar to the top form: 
\bea 
\mathbf{X}&=&X+Z^{\dagger}+\psi^{\dagger}+\sum_{s=1}^{n-1}\psi^{\dagger}_{(s)}\,, \\[4pt] 
\mathbf{A}&=& \epsilon+A+Y^{\dagger}+\chi^{\dagger}+\sum_{r=1}^{n-2}\chi^{\dagger}_{(r)}\,,\\[4pt] 
\mathbf{Y}&=& \sum_{r=1}^{n-2}\chi_{(n-r-1)}+\chi +Y+A^{\dagger}+\e^{\dagger}\,,\\[4pt] 
\mathbf{Z}&=& \sum_{s=1}^{n-1}\psi_{(n-s)}+\psi+Z+X^{\dagger}\,.
\eea 
If we denote the coordinates on $T[1]\S_{n+1}$ by $\s^{m}$ (the degree 0 ones) and $\theta^{m}$ (the degree 1 ones), then for example $Z^{\dagger}=Z^{\dagger}_{m}\theta^{m}$, $\psi^{\dagger}=\frac 12 \psi^{\dagger}_{mn}\theta^{m}\theta^{n}$ and so on.

Then the minimal solution to the classical master equation in the BV/BRST formalism of the R-Poisson sigma models is given as 
\bea 
&& S_{\text{BV}}=\int_{T[1]\S_{n+1}}\bigg(\mathbf{Z}_{\m}\dd \mathbf{X}^{\m}-\mathbf{A}_{\m}\dd \mathbf{Y}^{\m}  + \Pi^{\m\n}(\mathbf{X})\mathbf{Z}_{\m} \mathbf{A}_{\n} 
 \nn\\[4pt]  && \qquad\qquad - \, \frac 12 \partial_{\rho}\Pi^{\m\n}(\mathbf{X})\, \mathbf{Y}^{\rho} \mathbf{A}_{\m}\mathbf{A}_{\n}+\frac{1}{(n+1)!}R^{\m_1\dots\m_{n+1}}(\mathbf{X})\mathbf{A}_{\m_1}\dots\mathbf{A}_{\m_{n+1}}\bigg)\,. \label{RPSM BV}
\eea 
We note that this gives a solution only in the untwisted case. This is to be expected because one can show that in that case the target space is a QPn manifold.  In the twisted case the target is a Qn manifold and the minimal solution is not as simple to determine, see \cite{CIS} for techniques and results on the three-dimensional case. 

The expression \eqref{RPSM BV} is remarkably simple, which is the power of the AKSZ construction, but on the other hand it hides the structure of the theory with regard to the antifields and the interactions among them and the fields. When one expands the master action in antifields as{\footnote{There are no scalar antifields in the theory, therefore the ``longest'' antifield interaction contains $n+1$ of them, specifically the 1-form antifields $Z^{\dagger}$ multiplied by $n+1$ ghosts $\e$ to give a zero ghost $2n+2$ fermion interaction term.}} 
\be 
S_{\text{BV}}=\sum_{i=0}^{n+1}S_{(i)}\,,
\ee 
with $S_{(i)}$ containing $i$ antifields and $S_{(0)}$ being the classical action, the simplicity is apparently lost. The expansion has the advantage of highlighting the specific interaction vertices and also of allowing a direct and explicit calculation of the BRST differential, which consists of the Koszul-Tate differential, the longitudinal differential and the many additional terms in the spirit of homological perturbation theory \cite{HTbook}, since the model is in general highly reducible as a gauge theory. The coefficients that multiply all these interaction terms are scalar field dependent, partial derivatives of the basic structures $\Pi^{\m\n}$ and $R^{\m_1\dots\m_{n+1}}$. As expected, in general they do not transform linearly under target space coordinate transformations, even though the theory is of course diffeomorphism invariant by construction. Note, moreover, that the derivative of the Poisson 2-vector appears already in the classical action for higher dimensional Poisson sigma models.  

As explained in the introduction, one of our goals is to give a geometric meaning to these coefficients in terms of suitable geometric structures in ordinary differential geometric language. To this end we introduce an affine connection $\nabla$ on the cotangent bundle of the Poisson manifold such that 
\be 
\nabla_{\m}\Pi^{\k\l}=\partial_{\m}\Pi^{\k\l}+\G^{\k}_{\m\n}\Pi^{\n\l}+\G^{\l}_{\m\n}\Pi^{\k\n}\,,\label{nabla pi}
\ee 
and similarly for other quantities, where $\Gamma^{\m}_{\n\rho}$ are the connection coefficients. Precise details will appear in the following sections. We will only consider torsion free affine connections (more precisely this refers to the dual connection on the tangent bundle) for which $\Gamma^{\m}_{[\n\rho]}=0$. Clearly the covariant derivative in Eq. \eqref{nabla pi} is the tensor that corresponds to the nontensorial partial derivative of the components of the Poisson 2-vector. It is known that this tensor can be expressed as the torsion of an induced $T^{\ast}M$-connection on $T^{\ast}M$ \cite{Ikeda:2019czt}, as we will recall in the following. Aside the first derivative on the Poisson tensor, in general we encounter coefficients with multiple derivatives on it and also with derivatives on the coefficients of the multivector $R$. 

We would now like to explain what is the miminal set of coefficients we need to account for in order to have a geometrical description and a coordinate free formulation and which is the precise geometric form of these coefficients in local coordinates. First note that the homological vector field $Q^{\ast}$ in \eqref{Q Poisson} can be written in an alternative basis of derivations as 
\be 
Q^{\ast}=\Pi^{\m\n}a_{\m}\DD_{\n}-\frac 12 \nabla_{\m}\Pi^{\k\l}a_\k a_\l \frac{\partial}{\partial a_{\m}}\,,
\ee 
where 
\be 
\DD_{\n}=\frac{\partial}{\partial x^{\n}}+\Gamma^{\m}_{\n\rho}a_{\m}\frac{\partial}{\partial a_{\rho}}\,.
\ee 
Observe that in accord with the interpretation of the degree 1 fiber coordinate $a_{\m}$ as the graded geometric counterpart of the basis of the tangent bundle, we have
\be 
\DD_{\n}a_{\m}=\G_{\n\m}^{\rho}a_{\rho}\,.
\ee 
The homological vector field contains the information of the longitudinal differential. The information of the full BRST differential is contained in the Hamiltonian vector field of its Hamiltonian lift to the cotangent bundle. The Hamiltonian lift of $Q^{\ast}$ is 
\be 
h_{Q^{\ast}}=\Pi^{\m\n}a_{\m}x^{\ast}_{\n}-\frac 12 \partial_{\m}\Pi^{\k\l}a_{\k}a_{\l}a_{\ast}^{\m}\,,
\ee 
where $x^{\ast}_{\m}$ and $a_{\ast}^{\m}$ are the momenta on $T^{\ast}T^{\ast}[1]M$ of the same parity as the coordinates \cite{dimaphd}. The lift $h_{Q^{\ast}}$ is a smooth function on the cotangent bundle of the target space. Its Hamiltonian vector field is given as 
\bea 
X_{h_{Q^{\ast}}}=Q^{\ast}-\bigg(\partial_\k\Pi^{\m\n}a_{\m}x^{\ast}_{\n}-\frac 12 \partial_\k\partial_\l\Pi^{\m\n}a_{\m}a_{\n}a_{\ast}^{\l}\bigg)\frac{\partial}{\partial x^{\ast}_{\k}}+\bigg(\Pi^{\k\m}x^{\ast}_{\m}-\partial_{\m}\Pi^{\k\n}a_{\n}a_{\ast}^{\m}\bigg)\frac{\partial}{\partial a_{\ast}^{\k}}\,.
\eea 
This lifting procedure generates a second derivative on the Poisson 2-vector, which is known to arise in the master action and in the BRST differential of the 2D Poisson sigma model, as one would expect. 
Moreover, the tensor hiding behind this second derivative was shown in \cite{Ikeda:2019czt} to be the so-called basic curvature tensor of the connection on $T^{\ast}M$ \cite{Blaom,Crainic}. This can be also seen by expressing this Hamiltonian vector field in a suitable basis, but we skip this here because we will essentially repeat it for the twisted R-Poisson case. 

We may now ask what is the corresponding rewriting of the homological vector field $Q^{RP_n}_{H}$ of a twisted R-Poisson structure such that the expansion coefficients are tensorial quantities? Starting from expression \eqref{Q RPnH} and performing the covariantization due, we obtain 
\bea 
Q^{RP_n}_{H}&=& \Pi^{\m\n}a_{\m}\DD_{\n}-\frac 12 \nabla_{\rho}\Pi^{\m\n}a_{\m}a_{\n}\DD^{\rho} \nn\\[4pt] 
&& +\, \big((-1)^{n}\Pi^{\m\n}z^{\scriptscriptstyle\nabla}_{\m}-\nabla_{\m}\Pi^{\n\r}a_{\r}y^{\m}
+\frac 1{n!}R^{\n\m_1\dots\m_{n}}a_{\m_1}\dots a_{\m_{n}}\big)\dot\DD_{\n} \nn\\[4pt] 
&& +\, \bigg(\nabla_{\n}\Pi^{\m\r}a_{\r}z^{\scriptscriptstyle\nabla}_{\m}-\frac {(-1)^n}{2}\big(\nabla_{\n}\nabla_{\m}\Pi^{\r\s}-2\Pi^{\k[\r}{\mc R}^{\s]}{}_{\m\k\n}\big)y^{\m}a_{\r}a_{\s}\bigg)\dot{\DD}^{\n} \nn\\[4pt] 
&& - \, \frac{(-1)^n}{(n+1)!}\bigg(\nabla_{\n}R^{\m_1\dots \m_{n+1}}-H_{\n}{}^{\m_1\dots\m_{n+1}}\bigg)a_{\m_1}\dots a_{\m_{n+1}}\dot{\DD}^{\n}\,,
\eea 
where $z_{\m}^{\scriptscriptstyle\nabla}=z_{\m}+\G_{\m\n}^{\r}y^\n a_\r$ is the ``covariant momentum'' that transforms linearly under general coordinate transformations, ${\mc R}$ is the curvature of the connection $\nabla$ and the basis of derivations is now given as 
\bea 
\DD_{\m}&=&\frac{\partial}{\partial x^{\m}}+\G_{\m\n}^{\r}\big(a_{\r}\frac{\partial}{\partial a_{\n}}-y^{\n}\frac{\partial}{\partial y^{\r}}+z^{\scriptscriptstyle\nabla}_{\rho}\frac{\partial}{\partial z^{\scriptscriptstyle\nabla}_{\n}}\big)-\big(\partial_{\n}\G_{\m\r}^{\s}+\mc{R}^{\s}{}_{\r\m\n}\big)y^{\r}a_{\s}\frac{\partial}{\partial z^{\scriptscriptstyle\nabla}_{\n}}\,, \\[4pt] 
\DD^{\m}&=&\frac{\partial}{\partial a_{\m}}+(-1)^{n}\Gamma^{\m}_{\n\r}y^{\n}\frac{\partial}{\partial z^{\scriptscriptstyle\nabla}_{\r}}\,, \\[4pt] 
\dot{\DD}_{\m}&=&\frac{\partial}{\partial y^{\m}}-\Gamma_{\m\n}^{\r}a_{\r}\frac{\partial}{\partial z^{\scriptscriptstyle\nabla}_{\n}}\,, \\[4pt] 
\dot{\DD}^{\m}&=&\frac{\partial}{\partial z^{\scriptscriptstyle\nabla}_{\m}}\,,
\eea 
cf. \cite{Chatzistavrakidis:2023lwo} for a similar construction for general (4-form twisted) Courant algebroids for $n=2$.
As before, we can also compute the Hamiltonian lift, a function on the cubic order bundle $T^{\ast}T^{\ast}[n]T^{\ast}[1]M$ and its Hamiltonian vector field which will now also contain third derivatives on the components of the 2-vector $\Pi$ as well as second derivatives on the components of the $(n+1)$-vector $R$. This is the minimal set of data, corresponding to interaction coefficients in the field theory, that we would like to explain in geometric terms. Namely, we have the following six tensorial structures that pop out from the covariant expressions of the twisted R-Poisson sigma model in $n+1$ dimensions: 
\bea 
&& \Pi^{\m\n}\,, \quad \nabla_{\r}\Pi^{\m\n}\,, \quad \nabla_{\k}\nabla_{\l}\Pi^{\m\n}-2\Pi^{\r[\m}{\mc R}^{\n]}{}_{\l\r\k}\,, \\[4pt] 
&& R^{\m_1\dots \m_{n+1}}\,,\quad \nabla_{\n}R^{\m_1\dots\m_{n+1}}-H_{\n}{}^{\m_1\dots\m_{n+1}}\,, \\[4pt] 
&& \nabla_{\m}\nabla_{\n}R^{\m_1\dots\m_{n+1}}-\nabla_{\m}H_{\n}{}^{\m_1\dots\m_{n+1}}-(n+1)R^{\r[\m_1\dots\m_n}{\mc R}^{\m_{n+1}]}{}_{\n\r\m}\,. \label{basic poly in brst}
\eea 
In the following sections we explain how one can see those as generalised anchors and generalized torsion and basic curvature tensors on specific Lie $n$-algebroids.

\section{Twisted multibrackets on the cotangent bundle}
\label{sec3}

We briefly recall that a Lie algebroid is a triple $(E,b_2,\rho)$ of a vector bundle $E\overset{\pi}\longrightarrow M$ over a smooth manifold $M$, of a smooth bundle map $\rho:E\to TM$ that anchors sections $e\in \G(E)$ to vector fields on $M$ and of a skew-symmetric binary Lie bracket $b_2: \w^2\G(E)\to \G(E)$ (satisfying the  Jacobi identity) such that for a function $f\in C^{\infty}(M)$ it satisfies the Leibniz rule
\be \label{Lei LAoid}
b_2(e,fe')=f\,b_2(e,e')+L_{\rho(e)}f\, e'\,.
\ee 
The anchor map $\rho$ is, moreover, a homomorphism: 
\be \label{homoLAoid}
\rho(b_2(e,e')) =[\rho(e),\rho(e')]\,,
\ee 
where on the right hand side we encounter the ordinary Lie bracket of vector fields, a property that follows from the Jacobi identity and the Leibniz rule. Usually one uses the typical bracket notation{\footnote{Also in most of the literature the notation $\ell_2$ appears, with $b_2$ sometimes reserved for antialgebroids with degree suspension. We do not follow this convention here.}} $b_2(e,e')=[e,e']_{\E}$ but we shall refrain from doing so here in view of the fact that we are going to introduce higher $(n+1)$-ary brackets too and this notation will prove lighter. A Lie algebroid is transitive when the anchor map is fibrewise surjective. An almost Lie algebroid is a structure as above but such that in general the skew-symmetric bracket does not satisfy the Jacobi identity.   

In this paper we will primarily be interested in algebroid structures on the cotangent bundle $E=T^{\ast}M$ over some manifold $M$. A direct yet trivial example is the totally intransitive Lie algebroid where the bracket is identically zero for all pairs of 1-forms and the anchor sends every 1-form to the zero vector field. Even though uninspiring, this example can be useful in certain cases, when for instance one is interested in constructing examples of Lie bialgebroids and accordingly Courant algebroids, e.g. the so-called {standard} Courant algebroid. A more interesting and well known example arises when $M$ is a Poisson manifold with Poisson structure $\Pi\in\G(\w^2 TM)$. Then one can consider the anchor map induced by the Poisson structure, which we shall denote by the same symbol $\Pi: T^{\ast}M\to TM$, as long as no confusion arises.{\footnote{The notation $\Pi^{\sharp}$ is usually reserved for this map, but it should be clear from context when $\Pi$ is a 2-vector and when it is a map, so we will not insist on it.}} The Lie bracket on 1-forms is taken to be the Koszul bracket 
\be \label{KS}
b_2(e,e')={L}_{\Pi(e)}e'-{L}_{\Pi(e')}e-\dd (\Pi(e,e'))=\iota_{\Pi(e)}\dd e'-\iota_{\Pi(e')}\dd e+\dd(\Pi(e,e'))\,,
\ee  
where $\dd$ is the de Rham differential. This example can be directly generalized to the case when $M$ is not a genuine Poisson manifold but one twisted by a 3-form \cite{Severa Weinstein}. Recall that the triple $(M,\Pi,H)$ where $H\in\Omega_{\text{cl}}^3(M)$ (a closed 3-form) is a twisted Poisson manifold when 
\be \label{twisted Poisson}
\frac 12\, [\Pi,\Pi]_{\text{SN}}=\langle\Pi^{\otimes 3},H\rangle\,,
\ee 
where on the left hand side we encounter the Schouten-Nijenhuis bracket of multivector fields which  vanishes in the case of vanilla Poisson structure and $\Pi^{\otimes 3}$ is shorthand for $\Pi\otimes\Pi\otimes\Pi$ with the contractions understood in the first entries of the 2-vector; in local coordinates we obtain 
\be 
3 \Pi^{\m[\k}\partial_{\m}\Pi^{\l\n]}=\Pi^{\k'\k}\Pi^{\l'\l}\Pi^{\n'\n}H_{\k'\l'\n'}\,.
\ee 
We report here the local coordinate formula for the Schouten-Nijenhuis bracket, using the convention of \cite{Vaisman},
\bea \label{IzuV1} 
    [U,V]_{\text{SN}}\,&=
    \, & \frac{(-1)^{p}}{p!q!}\Big\{ q\,V^{\k\n_{2}...\n_{q}}\frac{\partial U^{\m_{1}...\m_{p}}}{\partial x^{\k}} \frac{\partial}{\partial x^{\m_{1}}} \wedge...\wedge \frac{\partial}{\partial x^{\m_{p}}} \wedge \frac{\partial}{\partial x^{\n_{2}}}\wedge ... \wedge \frac{\partial}{\partial x^{\n_{q}}} +
    \nn \\[4pt] &&\,   + \,(-1)^{p}p\,U^{\k\m_{2}...\m_{p}}\frac{\partial V^{\n_{1}...\n_{q}}}{\partial x^{\k}}\frac{\partial}{\partial x^{\m_{2}}}\wedge...\wedge \frac{\partial}{\partial x^{\m_{p}}}\wedge 
     \frac{\partial}{\partial x^{\n_{1}}}\wedge...\wedge \frac{\partial}{\partial x^{\n_{q}}}\Big\}\,,
\eea 
for $U\in \mf{X}^{p}(M)$ and $V\in \mf{X}^{q}(M)$ multivector fields of order $p$ and $q$ respectively, namely elements of the contravariant Grassmann algebra of $M$. For decomposable multivector fields  $U=X_{1}\wedge...\wedge X_{p}$ and $V=Y_{1}\w...\w Y_{q}$, a useful formula is
\be \label{IzuV2}
    [U,V]_{\text{SN}}= (-1)^{p+1}\sum_{r=1}^{p}\sum_{s=1}^{q}(-1)^{r+s}[X_{r},Y_{s}]\w U[r]\w V[s]\,,
\ee
where we introduced the notation $U[r]$ to mean the decomposable multivector with exclusion of the $r$-th entry.
Applying \eqref{IzuV2} to multivector fields of local form $N=\frac{1}{n!}N^{\m_{1}...\m_{n}}\frac{\partial}{\partial x^{\m_{1}}}\wedge...\wedge\frac{\partial}{\partial x^{\m_{n}}}$, yields \eqref{IzuV1}.

For vanishing $H$ the Lie algebroid structure on $T^{\ast}M$ was described above; for nonvanishing $H$ the Lie algebroid structure is again given with the anchor $\Pi$ and with the Lie bracket being the \emph{twisted} Koszul one: 
\be 
b_2^{\H}(e,e')=b_2(e,e')-H\big(\Pi(e),\Pi(e'),\cdot\big)\,.
\ee 
That this is a Lie algebroid follows from the defining relations \eqref{twisted Poisson} of a twisted Poisson structure. We emphasize that the 3-form twist that obstructs the Poisson-ness does not take us outside the category of Lie algebroids. This is in contrast to Courant algebroids, where a 4-form twist leads to the relaxed structure of a pre-Courant algebroid \cite{Vaisman:2004msa}.

In this work we shall be interested in structures on $T^{\ast}M$ that go beyond the above considerations. Specifically, we shall relax the assumption that $M$ is just a Poisson or twisted Poisson manifold and that the structure on its cotangent bundle is a Lie algebroid. This is motivated by the construction of the explicit realizations of higher-dimensional Hamiltonian mechanics that we described in Section \ref{sec2} and it will be achieved via the introduction of higher brackets. We proceed constructively and we shall return to the discussion of what is the precise structure on the manifold $M$ in due course. To define such structures we first consider the exterior power bundle of $n$-forms on the manifold $M$, which we denote as $E_n:=\w^n T^{\ast}M$. Its sections are differential forms of degree $n$ that we denote as $\hat{e}\in\G(E_n)=\Omega^{n}(M)$ for brevity.
Moreover, we promote $E_n$ to an anchored vector bundle $(E_n,R)$ by introducing the map 
\be 
R: E_n \to TM\,.
\ee 
With some abuse of notation as earlier, we can alternatively think of this map as a degree $n+1$ fully antisymmetric multivector field $R\in \mf{X}^{n+1}(M)\simeq\G(\w^{n+1}TM)$. Such objects appear naturally in the context of homotopy Poisson structures on manifolds, also called $P_{\infty}$ structures \cite{Voronov2,Voronov,Cattaneo:2005zz}. In fact much of our discussion in the present section parallels \cite{Voronov2}.

With the aid of the map $R$ introduced above, we define an $(n+1)$-ary bracket as a map $b_{n+1}: \w^n\G(E) \to \G(E)$ given as follows: 
 \bea \label{higherKoszul}
 b_{n+1}(e^{(1)},\dots, e^{(n+1)})=\sum_{r=1}^{n+1}(-1)^{n-r+1} {L}_{R({\hat e}[r])}e^{(r)}-\frac 1{(n-1)!}\, \dd\big(R(e^{(1)},\dots , e^{(n+1)})\big)\,, 
 \eea 
 where $e^{(r)}\in\G(T^{\ast}M)$ are 1-forms and $\hat e[r]\in \G(E_n)$ is the decomposable $n$-form given by the exclusion of the $r$-th entry in the wedge string of $e^{(r)}$s, namely{\footnote{The notation $[\cdot]$ should not be confused with the shift functor. It only applies to sections here.}} 
 \be \label{deco}
 \hat e[r]=e^{(1)}\w\dots e^{(r-1)}\w e^{(r+1)}\w\dots \w e^{(n+1)}\,.
 \ee  
In practice it is sometimes useful for computations to express the bracket in terms of interior products only, as we did for the binary one in Eq. \eqref{KS}, that is 
\be 
 b_{n+1}(e^{(1)},\dots, e^{(n+1)})=\sum_{r=1}^{n+1}(-1)^{n-r+1} \iota_{R({\hat e}[r])}\dd e^{(r)}+\frac 1{n!}\, \dd(R(e^{(1)},\dots , e^{(n+1)}))\,.
\ee 
As a technical remark, note the change in sign and factor of the second term on the right hand side in this alternative expression. This change in the factor is obviously not visible in the typical $n=1$ case.  In this form it is easy to see that in a local coordinate basis the multibracket becomes
\bea 
&& b_{n+1}(e^{(1)},\dots, e^{(n+1)})=\bigg(\sum_{r=1}^{n+1}\frac{(-1)^{n-r+1}}{n!}R^{\m_1\dots\m_{n+1}}e^{(1)}_{\m_1}\dots e^{(r-1)}_{\m_{r-1}}e^{(r+1)}_{\m_{r}}\dots e^{(n+1)}_{\m_n}\partial_{\m_{n+1}}e^{(r)}_{\l}+\nn\\[4pt] && \qquad\qquad\qquad \qquad\qquad\quad\quad +\frac 1{n!}\partial_{\l}R^{\m_1\dots\m_{n+1}}e^{(1)}_{\m_1}\dots e^{(n+1)}_{\m_{n+1}}\bigg)\dd x^{\l}\,,
\eea 
and as a result that the local coordinate formula
\be 
b_{n+1}(\dd x^{\m_1},\dots,\dd x^{\m_{n+1}})_{\m}=\frac 1{n!}\partial_{\m}R^{\m_1\dots\m_{n+1}}
\ee 
holds. We thus see that the derivative on the multivector field can be obtained as a structure constant of the multibracket. Of course we are still missing the components of the $(n+2)$-form $H$, which indicates that we must twist the higher bracket.

 Note that for $n=1$, in which case $E_{n=1}=T^{\ast}M$ and under the identification of the 2-vector $R$ with a Poisson structure $\Pi$, the bracket in \eqref{higherKoszul} reduces to the usual binary Koszul bracket of 1-forms appearing in Eq. \eqref{KS}. We note that in this case the cyclic sum of Lie derivatives contains the expected two terms of opposite sign. To appreciate the structure, let us jump to the next level and write down explicitly the form of the ternary bracket for $n=2$. It is 
 \be \label{higherKoszul3}
 b_3(e^{(1)},e^{(2)}, e^{(3)})={ L}_{R(e^{(2)}\w e^{(3)})}e^{(1)}+{ L}_{R(e^{(3)}\w e^{(1)})}e^{(2)}+{ L}_{R(e^{(1)}\w e^{(2)})}e^{(3)} - \dd\big(R(e^{(1)},e^{(2)} , e^{(3)})\big)\,, 
 \ee
with $R$ being an antisymmetric 3-vector and we observe again the cyclicity in the Lie derivatives.

It is straightforward to see that the $(n+1)$-ary Koszul bracket is antisymmetric in all its entries, i.e. 
\be 
b_{n+1}(e^{(1)},\dots,e^{(r)},\dots,e^{(s)},\dots,e^{(n+1)})=-\,b_{n+1}(e^{(1)},\dots,e^{(s)},\dots,e^{(r)},\dots,e^{(n+1)})\,.
\ee 
This follows directly from the definition. 
Furthermore, it is also straightforward to show that it is not $C^{\infty}(M)$-linear but instead it satisfies the Leibniz rule for functions $f\in C^{\infty}(M)$,
\be 
b_{n+1}(fe^{(1)},e^{(2)},\dots,e^{(n+1)})=f\,b_{n+1}(e^{(1)},e^{(2)},\dots, e^{(n+1)})+R(\dd f,e^{(2)},\dots, e^{(n+1)})\, e^{(1)}\,,
\ee 
extending to all other entries of the bracket by full antisymmetry. The second term in this Leibniz rule may as well be written in terms of the anchor map $R$, denoted by the same symbol, as $(-1)^{n}R(e^{(2)}\wedge \dots \wedge e^{(n+1)})f\, e^{(1)}$, or equivalently in terms of the Lie derivative,
\be 
b_{n+1}(fe^{(1)},e^{(2)},\dots,e^{(n+1)})=f\,b_{n+1}(e^{(1)},e^{(2)},\dots, e^{(n+1)})+(-1)^{n}L_{R(e^{(2)},\dots, e^{(n+1)})}f\, e^{(1)}\,.
\ee 
In this form it is simpler to see that for $n=1$ it reduces to the usual rule \eqref{Lei LAoid} for Lie algebroid brackets. 
 \begin{rmk}
Recall that the Poisson 2-vector $\Pi$ and the usual Poisson bracket of functions on phase space are related as 
 \be 
 \{f,g\}_{\text{P.B.}}=\Pi(\dd f,\dd g)\,
 \ee 
 and that the (untwisted) Koszul bracket is such that on exact 1-forms it is given by the de Rham differential of the Poisson bracket:
 \be 
 \dd\{f,g\}_{\text{P.B.}}=b_2(\dd f,\dd g)\,.
 \ee 
 The same logic holds true in the higher case, in particular given a $(n+1)$-vector $R$ we can define a higher Poisson bracket of functions on phase space as 
 \be 
 \{f_1,\dots, f_{n+1}\}_{\text{P.B.n}}=R(\dd f_1,\dots \dd f_{n+1})\,,
 \ee 
  and moreover the $(n+1)$-ary Koszul bracket is related to this higher Poisson bracket as 
  \be 
  \dd\{f_1,\dots,f_{n+1}\}_{\text{P.B.n}}=b_{n+1}(\dd f_1,\dots,\dd f_{n+1})\,,
  \ee 
  which is the direct generalization of the $n=1$ case. One can give a derived bracket perspective to these relations in the context of homotopy Poisson structures \cite{Voronov}. We also note that in the terminology of \cite{Ibanez}, this is an almost Poisson bracket of order $n+1$, the multivector $R$ is an almost Poisson $(n+1)-$tensor and the manifold $(M,R)$ is a generalized almost Poisson manifold. 
 \end{rmk}

 There are two general ways that one can make the above ``almost'' structure more specific. The first is to impose the additional condition  
 \be  
[R,R]_{\text{SN}}=0\,.
 \ee  
 Taking into account the antisymmetry of the Schouten-Nijenhuis bracket, this condition is vacuous (trivially satisfied) for odd $n+1$ but contentful for even $n+1$. We nevertheless refer to both cases as generalized Poisson. When written in terms of the almost Poisson bracket of order $n+1$ this condition is identical to a completely cyclic fundamental identity (which replaces the Jacobi identity and it reduces to it for $n=1$). The second option is to impose a different fundamental identity, one that turns the corresponding Hamiltonian vector field into a derivation of the algebra of functions generated by the almost Poisson bracket of order $n+1$. This is what is called a Nambu-Poisson structure and it has been used in a different way than here to construct actions for higher dimensional branes \cite{Jurco:2012gc}. In fact, in most of the ensuing we shall deal with general almost Poisson structures of order $n+1$ in the above sense; therefore we do not present more details on these special cases here. We just note that there exists a notion of $n$-Lie algebroid, generalizing $n$-Lie algebras which are relevant when $M$ is a point. 

 One can now ask whether the higher bracket $b_{n+1}$ can be twisted by higher degree closed differential forms, as was the case for the binary bracket of the Lie algebroid on the twisted Poisson manifold. One option is to consider the $(n+1)$-ary bracket as twisted through ``its own'' structural quantity $R$. However, $R$ acts on an $n$-form to produce a vector field which should in turn be acted upon by a 2-form to recreate a 1-form as a result. This may be achieved by considering that the base manifold $M$ has torsion. In other words that it is equipped with an affine connection $\nabla$ with torsion $T$, which is a vector valued 2-form.{\footnote{This is not what will be useful for our purposes, since as the reader recalls from Section \ref{sec2} the affine connection we considered there does not have torsion. However, it is useful to discuss this option in order to make contact with the concept of twisted Poisson structure introduced earlier in this section.}} Then it is possible to define the following first version of twisted $(n+1)$-ary bracket, which we call the torsion-twisted (or in brief $T$-twisted) $(n+1)$-ary Koszul bracket: 
\be 
b_{n+1}^{\textit{\tiny{T}}}=b_{n+1}-\mathrm{Tw}_{n+1}\,,
\ee 
where the 1-form twist is defined as 
\be 
 \mathrm{Tw}_{n+1}(e^{(1)},\dots,e^{(n+1)})=\frac 1{(n+1)!} \sum_{r=1}^{n+1}(-1)^{n-r}\langle e^{(r)},T\rangle(R(\hat{e}[r]),\cdot)\,,\label{twist}
 \ee 
 and the angle brackets denote the canonical pairing between the tangent and cotangent bundles on $M$.
 \begin{rmk}\label{Tw n=1}
 For $n=1$ and in case of a twisted Poisson manifold $(M,\Pi,H)$ with a closed 3-form $H$, one may consider the connection $\nabla$ with torsion 
 \be 
 T=\langle\Pi,H\rangle\,.
 \ee 
 Then the twist \eqref{twist} reduces to 
 \be 
\mathrm{Tw}_2(e^{(1)},e^{(2)})=H\big(\Pi(e^{(1)}),\Pi(e^{(2)}),\cdot\big)\,,
 \ee 
 which is the usual one of the 3-form-twisted (binary) Koszul bracket. In other words, we find that with the above definitions there is an identification $b_2^{\H}=b_2^{\textit{\tiny{T}}}$. 
 \end{rmk}

 The above twist is not, however, the only option when additional structures are present alongside the bracket's ``own'' structure $R$. Suppose for example that we have a manifold equipped with both a 2-vector $\Pi$ and an $(n+1)$-vector $R$ with $n\ge 1$ and moreover it is endowed with a closed $(n+2)$-form $H\in \Omega_{\text{cl}}^{n+2}(M)$. This is clearly the case that corresponds to the field theories of Section \ref{sec2}. Consider, moreover, that $\Pi$ is a genuine Poisson 2-vector. Then on $T^{\ast}M$ we can define both binary and $(n+1)$-ary Koszul brackets. The binary bracket will be the untwisted one---unless of course one introduces an additional closed 3-form, which we do not do. However, the $(n+1)$-ary bracket can now be twisted even without any torsion on the manifold $M$. 
 This $(n+2)$-form-twisted (or in brief $H$-twisted) $(n+1)$-ary Koszul bracket is defined as 
     \be  
b_{n+1}^{\H}=b_{n+1}-\frac{1} {n!}H_{n+1}\,,
     \ee  
     where the 1-form twist is given by 
     \be \label{Hn+1}
H_{n+1}(e^{(1)},\dots,e^{(n+1)})=H\big(\Pi(e^{(1)}),\dots,\Pi(e^{(n+1)}),\cdot\big)\,.
     \ee 
     We observe that for the choice of connection with torsion on $M$ as in Remark \ref{Tw n=1}, it holds that $Tw_2=H_2$, which also justifies the notation since then $b_2^{\H}=b_2^{\textit{\tiny{T}}}$. The two twists cannot agree for any other value of $n>1$.
 For future reference and to keep a record of our conventions we also give the local coordinate form of the $H$-twisted higher bracket:
 \be 
b_{n+1}^{\H}(\dd x^{\m_1},\dots,\dd x^{\m_{n+1}})_{\m}=\frac 1{n!}(\partial_{\m}R^{\m_1\dots \m_{n+1}}-H_{\m}{}^{\m_1\dots\m_{n+1}})\,.
 \ee 
We observe that the structure constants of the $H$-twisted bracket generate one of the terms we encountered in the non manifestly covariant expressions of the  gauge theory discussed in Section \ref{sec2}.

 So far, the only properties we have mentioned for the $(n+1)$-ary bracket are skew-symmetry and the Leibniz rule. Much like the binary one, it might satisfy some Jacobi-like identity and some property analogous to the anchor $\Pi$ being a homomorphism. In the next section we will study precisely these properties. At this stage we can be more specific and focus hence on the twisted R-Poisson structure, introduced in \cite{Thanasis R-Poisson}. A twisted R-Poisson manifold is the quadruple $(M,\Pi,R,H)$ of structures as they appear above with the condition 
\be 
[\Pi,R]=(-1)^{n+1}\langle\Pi^{\otimes(n+2)},H\rangle\,.
\ee
When $H=0$ we call it an R-Poisson manifold and if in addition $R=0$ it becomes a genuine Poisson manifold. To avoid confusion we note that a twisted Poisson manifold in the sense of \v{S}evera and Weinstein is not a special case of twisted R-Poisson manifold for any choice of structural data.{\footnote{It is however a special case of \emph{bi-twisted} R-Poisson manifold \cite{Thanasis R-Poisson}.}} 

\paragraph{Higher Schouten-Nijenhuis brackets.} 

To complete our analysis in this section, it is useful to recall that given a bracket on $E$ we can always equip $E_n$ with the natural binary bracket that generalizes the Schouten-Nijenhuis bracket of multivectors and turns the space of multisections into a Gesternhaber algebra \cite{MX}. To keep the notation minimal we again denote this bracket as $b_2:\G(E_m)\times \G(E_n)\to \G(E_{m+n-1})$. For $\hat e \in \G(E_m)$ and $\hat e'\in\G(E_n)$ it is given as 
\be  
b_2(\hat e,\hat e')=\sum_{r,s}(-1)^{r+s}b_2(e^{(r)},e'^{(s)})\w \hat{e}[r]\w\hat{e}'[s]\,.
\ee 
It satisfies graded antisymmetry, graded derivation rule and graded Jacobi identity as laid out in section 2 of Ref. \cite{MX} or Chapter 7 of \cite{Mackenzie}, and it agrees with the binary bracket on $E$ and with the anchor $\rho$ when one of the entries is a function of $C^{\infty}(M)$. 
For example the graded derivation rule reads
\be 
b_2(\hat{e},\hat{e}'\w\hat{e}'')=b_2(\hat{e},\hat{e}')\w\hat{e}''+(-1)^{(\text{deg}(\hat{e})-1)\text{deg}(\hat{e}')}\hat{e}'\w b_2(\hat{e},\hat{e''})\,.
\ee 
The Leibniz rule for functions $f\in C^{\infty}(M)$ according to the above formula reads
\be 
b_2(\hat{e},f\hat{e}')=f\, b_2(\hat{e},\hat{e}')+L_{\rho(X)}f \,\hat{e}'\,,
\ee 
exactly as the original binary bracket on sections of $E$.
The Schouten bracket permeates the full cotangent complex generated by the differential $\dd_{E}$, which is defined by the usual Koszul formula like the de Rham differential. 

In the context of our discussion, there is more to the Schouten algebra than the binary bracket given above. There is also the option of taking the $(n+1)$-ary bracket for multisections. Aligning with our previous notation, we write this bracket 
\be
b_{n+1}:\G(E_{m_1})\times\dots \G(E_{m_{n+1}})\to \G(E_{\bar{m}})\,,\quad \bar{m}=\sum\limits_{i=1}^{n+1}m_i-n\,,
\ee 
through the following formula
\be \label{Schouten n+1}
b_{n+1}(\hat{e}^{(1)},\dots,\hat{e}^{(n+1)})=\sum_{r_i}(-1)^{p}b_{n+1}(e^{{(1)}^{(r_1)}},\dots,e^{{(n+1)}^{(r_{n+1})}})\w\hat{e}^{(1)}[r_1]\w\dots\w\hat{e}^{(n+1)}[r_{n+1}]\,,
\ee 
where 
\be 
p=n+1+\sum\limits_{i=1}^{n+1}r_i\,,
\ee 
in terms of the rather evident notation that the decomposable $n$-form $\hat{e}^{(1)}$ is given as $\hat{e}^{(1)}=e^{{(1)}^{(1)}}\w e^{{(1)}^{(2)}}\w\dots\w e^{{(1)}^{(m_1)}}$ 
and so on. Note that the $(n+1)$-ary Schouten-Nijenhuis bracket of an $n$-form and $n$ 1-forms returns an $(n+1)$-form. In general, this $(n+1)$-ary  bracket has degree $-n$, generalizing the degree $-1$ binary Schouten-Nijenhuis bracket. Note that for a function $f\in C^{\infty}(M)$ it holds that 
\bea 
&& b_{n+1}(f\hat{e}^{(1)},\dots,\hat{e}^{(n+1)})=f\, b_{n+1}(\hat{e}^{(1)},\dots,\hat{e}^{(n+1)})+ \nn\\[4pt]   &+& \,\sum_{r_i[1]}(-1)^{p[1]+1+n}R(e^{{(2)}^{(r_2)}}\w\dots\w e^{{(n+1)}^{(r_{n+1})}})f \,\hat{e}^{(1)}\w\hat{e}^{(2)}[r_2]\w\dots\w\hat{e}^{(n+1)}[r_{n+1}]\,,\nn\\
\eea 
where the sum is over all $r_i$ apart from $r_1$ and similarly $p[1]$ means that the sum in $p$ starts from 2 instead of 1. We observe that there is more to the Leibniz rule in the $(n+1)$-ary Schouten bracket than the original one except of course in the case $n=1$ when all sections are just forms on $E_1$ and the extra sum is trivial.

\section{Gapped (almost) Lie n-algebroids}\label{sec4}

 We already discussed that given a Poisson or twisted Poisson manifold there exists a Lie algebroid on its cotangent bundle. This begs the question: \emph{how can we describe a twisted R-Poisson manifold as an algebroid structure?} To answer this question,\footnote{See also \cite{Ikeda:2023vfq} for a different, alternative analysis.} first we find that the following identity holds: 
 \bea  
&& \Pi\big(b_{n+1}^{\H}(e^{(1)},\dots,e^{(n+1)})\big)+\sum_{r=1}^{n+1}(-1)^{n+1}R\big(b_2(e^{(r)},e^{(r+1)})\w \hat{e}[r,r+1]\big)= \nn\\ && \qquad = \sum_{r=1}^{n+1}(-1)^{n+1-r}[R(\hat e[r]),\Pi(e^{(r)})]\,, \label{mixed morphism}
 \eea 
 where $\hat{e}[r,r+1]=(\hat{e}[r])[r+1]$, establishing a relation between the binary and $(n+1)$-ary $H$-twisted Koszul brackets and the Lie bracket of vector fields by means of the maps $\Pi$ and $R$ that act on sections of $E_1$ and $E_n$ respectively and return a vector field on $M$. 
 Since we have left the almost Poisson structure $R$ general and not further specified, there is no additional Jacobi or fundamental identity for the $(n+1)$-ary bracket alone, at least for even $n$. However, we can search for one that involves both $b_2$ and $b_{n+1}$ for every $n$, even or odd. In the case of vanishing $H$, this mixed Jacobi identity reads 
 \be\label{mixed jacobi}
\text{Alt}\bigg(b_2\big(b_{n+1}(e^{(1)},\dots,e^{(n+1)}),e^{(n+2)}\big)+(-1)^{n+1} b_{n+1}\big(b_2(e^{(1)},e^{(2)}),e^{(3)},\dots,e^{(n+2)}\big)\bigg)=0\,,
 \ee
 where $\text{Alt}$ denotes antisymmetrization with weight 1 in all $n+2$ entries.  Notably, this identity is \emph{not} satisfied for $b_{n+1}^{\H}$, a fact that points to the existence of an almost structure, similar to the definition of an almost Lie algebroid that we mentioned earlier.
To reach a final conclusion we should examine whether property \eqref{mixed morphism} follows from the Leibniz rules for the brackets $b_2$ and $b_{n+1}$ together with the Jacobi identity for $b_2$ and the mixed Jacobi identity \eqref{mixed jacobi}. 
Anticipating that this is indeed the case and to also include bundles other than the cotangent we first give the following general definition. 
\begin{defn}\label{gapped def}
    A \textbf{gapped Lie n-algebroid} for $n>1$ is the quintet $(E,b_2,b_{n+1},\rho_2,\rho_{n+1})$ consisting of a vector bundle $E$ over a smooth manifold $M$, two smooth maps $\rho_2: E\to TM$ and $\rho_{n+1}: E_n:=\w^nE\to TM$ and two skew-symmetric operations on the space of sections $\G(E)$, one binary Lie bracket $b_2$ and one $(n+1)$-ary bracket $b_{n+1}$ that satisfy the Leibniz rules 
\bea 
\label{Lei 1} b_2(e,fe')&=& f b_2(e,e')+\rho_{2}(e)f\,e'\,, 
\\[4pt] 
\label{Lei 2} b_{n+1}(fe,e^{(1)},\dots,e^{(n)})&=&fb_{n+1}(e,e^{(1)},\dots,e^{(n)})+(-1)^{n}\rho_{n+1}(e^{(1)}\w\dots\w e^{(n)})f\, e\,,\,\,\,\,\,\,\quad
\eea 
for $f\in C^{\infty}(M)$ and the mixed Jacobi identity \eqref{mixed jacobi}. If instead the identity \eqref{mixed jacobi} is violated, then we speak of a \textbf{gapped almost Lie n-algebroid} if the mixed morphism condition \eqref{mixed morphism} (understood with the replacements $\Pi\to \r_2$ and $R\to \r_{n+1}$) is satisfied. We call the number $n-2$ the gap of a gapped Lie n-algebroid.
\end{defn}
Note that our definition of a gapped Lie $n$-algebroid is a special case of Definition 1 for a $L_{\infty}$ algebroid in \cite{Voronov2}. Indeed, in that work Th. Voronov demonstrates that the cotangent bundle equipped with a series of higher Koszul brackets is an $L_{\infty}$ algebroid. Here we also included the twist that gives rise to the almost structure. The usefulness of the concept for our purposes becomes transparent through the following example.
\begin{exa}
\label{example of gapped}
A nontrivial example of a gapped almost Lie n-algebroid is given by a twisted R-Poisson manifold, with $b_2$ and $b_{n+1}$ the binary and twisted $(n+1)$-ary Koszul brackets, $\rho_2=\Pi$ and $\rho_{n+1}=R$. We call this structure a \textbf{twisted R-Poisson $n$-algebroid}. In the untwisted case, i.e. when the $(n+2)$-form $H$ vanishes, the corresponding R-Poisson $n$-algebroid is an example of a gapped Lie n-algebroid. 
When $R=0=H$ the twisted R-Poisson $n$-algebroid reduces to the usual Poisson Lie algebroid structure on the cotangent bundle.
\end{exa}
  Let us make two important remarks. First, a gapped almost Lie $n$-algebroid contains a substructure $(E,b_2,\rho_2)$ which is a genuine Lie algebroid. To highlight this we call this the \textbf{ground Lie algebroid} of a gapped almost Lie $n$-algebroid. In the twisted R-Poisson example, the ground Lie algebroid is obviously the usual Lie algebroid on a Poisson manifold described earlier. 
Second, as already mentioned a gapped Lie n-algebroid is a Lie n-algebroid in the sense of Definition 1 in Ref. \cite{Voronov2}  with all brackets in the gap (from the 3-bracket to the $n$-bracket) vanishing, unless of course $n=2$ when the gap is absent. This justifies its name. On the other hand there exist different constructions in the literature that provide a one-to-one correspondence between homological vector fields and (split) $L_{\infty}$ algebroids defined with the requirement that all brackets higher than the binary one are $C^{\infty}(M)$-linear \cite{SZ,BP,LGLS}---see however \cite{herbig}. Although in the present case we also have a homological vector field associated to a gapped Lie $n$-algebroid, as we will  discuss in more details in the ensuing, and therefore there does exist a construction of derived brackets that yield a $C^{\infty}(M)$-linear $(n+1)$-ary bracket, the construction with the higher Koszul brackets that we use in this paper  is different and it also yields an $L_{\infty}$ algebroid with non $C^{\infty}(M)$-linear brackets in the same way as in \cite{Voronov2}.   

Due to the ground Lie algebroid structure, the map $\rho_2=\Pi$ of a twisted R-Poisson $n$-algebroid is a homomorphism for the Lie algebra structures on the space of sections of the cotangent and tangent bundles given by $b_2$ and the Lie bracket of vector fields; it satisfies \eqref{homoLAoid}. An additional property of the twisted R-Poisson $n$-algebroid, as announced earlier, is captured by the following:
\begin{prop}
The anchor maps $\Pi$ and $R$ of a twisted R-Poisson $n$-algebroid satisfy identity \eqref{mixed morphism} and, more generally, for any gapped Lie $n$-algebroid with the replacements $\rho_2$ and $\rho_{n+1}$ for $\Pi$ and $R$ respectively, identity \eqref{mixed morphism}  follows from Definition \ref{gapped def}.
\end{prop}
\begin{proof} It is sufficient to show that the mixed morphism property \eqref{mixed morphism} follows from the Leibniz rules and the Jacobi identities, which are part of the definition of a gapped Lie $n$-algebroid. In other words to perform the analogous computation that for Lie algebroids demonstrates that the anchor being a homomorphism follows from the Leibniz rule and the Jacobi identity for the Lie bracket on its space of sections \cite{KS}. First we use the Leibniz rules \eqref{Lei 1} and \eqref{Lei 2} to directly write 
\bea 
\rho_2\big(b_{n+1}(e^{(1)},\dots,e^{(n+1)})\big)f\,e^{(n+2)}&=&b_2\big(b_{n+1}(e^{(1)},\dots,e^{(n+1)}),fe^{(n+2)}\big)- \nn\\[4pt] 
&& \qquad - \, f\,b_2\big(b_{n+1}(e^{(1)},\dots,e^{(n+1)}),e^{(n+2)}\big)\,,\label{Lei 1 b}
\eea 
and, respectively, for $r=1,\dots,n+1$, 
\bea 
&& 
(-1)^n\rho_{n+1}\big(b_2(e^{(r)},e^{(r+1)})\w\hat{e}[r,r+1]\big)f\,e^{(n+2)} = \nn\\[4pt] 
&& \qquad\qquad\qquad b_{n+1}\big(f e^{(n+2)},b_2(e^{(r)},e^{(r+1)}),e^{(1)},\dots,e^{(n+1)}\big)- \nn\\[4pt] 
&& \qquad \qquad \qquad -\, f\,b_{n+1}\big(e^{(n+2)},b_2(e^{(r)},e^{(r+1)}),e^{(1)},\dots,e^{(n+1)}\big)\,,
\eea 
where the $r$-th and $(r+1)$-st entries in the ellipses of the terms on the right hand side of this formula are obviously excluded. In fact, the second formula is redundant for the proof that follows, but it is useful to have it on record nevertheless. We can then start from \eqref{Lei 1 b} and compute, using identity \eqref{mixed jacobi}:
\bea 
&& \rho_2\big(b_{n+1}(e^{(1)},\dots,e^{(n+1)})\big)f\, e^{(n+2)}
= \nn\\[4pt] 
&&   \qquad =(-1)^nb_2\big(b_{n+1}(fe^{(n+2)},e^{(1)},\dots,e^{(n)}),e^{(n+1)}\big) \nn\\[4pt] 
&& \qquad -\, b_2\big(b_{n+1}(e^{(n+1)},f e^{(n+2)},e^{(1)},\dots, e^{(n-1)}),e^{(n)}\big) +\dots \nn\\[4pt] && 
\qquad \dots +\, (-1)^{n}b_2\big(b_{n+1}(e^{(2)},\dots,f e^{(n+2)}),e^{(1)}\big) \nn\\[4pt]
 && \qquad
+\,(-1)^{n}b_{n+1}\big(b_2(e^{(1)},e^{(2)}),e^{(3)},\dots, f e^{(n+2)}\big) \nn\\[4pt] 
&& \qquad -\, b_{n+1}\big(b_2(f e^{(n+2)},e^{(1)}),e^{(2)},\dots, f e^{(n+1)}\big)+\dots \nn\\[4pt] 
&& \qquad \dots - \,b_{n+1}\big(b_2(e^{(2)},e^{(3)}),e^{(4)}
\dots,f e^{(n+2)},e^{(1)}\big)\nn\\[4pt] 
&& \qquad -f b_2\big(b_{n+1}(e^{(1)},\dots,e^{(n+1)}),e^{(n+2)}\big)\,,
\eea 
where the ellipses run over the cyclic permutations of the entries.
The next step is to use the Leibniz rules in all terms, in some cases once and in some cases twice. Schematically, for terms of the form $b_{n+1}(b_2(e,e),fe,\dots,e)$ we use the Leibniz rule \eqref{Lei 2}, whereas for terms of the form $b_2(b_{n+1}(fe,e,\dots,e),e)$ or $b_{n+1}(b_2(e,fe),e,\dots e)$ we first use the corresponding Leibniz rule for the inner bracket and then the other Leibniz rule for the outer bracket. Performing this task, we end up with the following classes of terms: (i) $f$ multiplying all the terms in identity \eqref{mixed jacobi}, thus vanishing, (ii) terms of the form $\r_2(e)f\, b_{n+1}(e,\dots,e)$ and $\r_{n+1}(e\w\dots\w e)f\,b_2(e,e)$ that all cancel among themselves, (iii) terms of the form $\rho_{n+1}(b_2(e,e)\w\dots\w e)f\,e$ which collect into the second term on the left hand side of \eqref{mixed morphism}, and (iv) terms quadratic in anchors, namely $\rho_2(e)\rho_{n+1}(e\w\dots\w e)f \,e^{(n+2)}$ and their antisymmetric partners $-\rho_{n+1}(e\w\dots\w e)\rho_2(e)f \,e^{(n+2)}$ that collect into the Lie derivative of vector fields on the right hand side of \eqref{mixed morphism}.
\end{proof}

A gapped Lie $n$-algebroid may also be described as an NQ$n$ manifold. 
The conditions of Definition \ref{gapped def} may be encoded in a homological vector field by considering  dg manifolds of degree $n$ arising from double vector bundles $(\widehat{E}=T[n-1]E[1], Q^{\Eh})$. In local coordinates $(x^{\m},a^{a},\dot{x}^{\m},\dot{a}^{a})$ of respective degrees $(0,1,n-1,n)$, the homological vector field is of the form
    \bea 
Q^{\widehat\E}&=&Q^{\E}+\big(\rho_{a}{}^{\m}\dot{a}^{a}+\partial_{\n}\rho_{a}{}^{\m}a^{a}\dot{x}^{\n}+\frac{1}{n!}\rho_{a_1\dots a_{n}}{}^{\m}a^{a_1}\dots a^{a_{n}}\big)\frac{\partial}{\partial \dot{x}^{\m}} \nn\\[4pt] 
&&\quad -\,\big(C^{a}_{bc}a^{b}\dot{a}^{c}+\frac 12 \partial_{\m}C^{a}_{bc}\dot{x}^{\m}a^{b}a^{c}+\frac 1{(n+1)!}C_{a_1\dots a_{n+1}}^{a}a^{a_1}\dots a^{a_{n+1}}\big)\frac{\partial}{\partial \dot{a}^{a}}\,,
    \eea 
where $Q^{\E}$ is the homological vector field of the ground Lie algebroid. 

 The correspondence follows via the identifications of the various components of the vector field with the coefficients of the structural data of a gapped Lie $n$-algebroid, specifically 
\be 
\rho_2=(\rho_{a}{}^{\m})\,,\quad \rho_{n+1}=(\rho_{a_1\dots a_n}{}^{\m})\,, \quad b_2=(C_{ab}^{c})\,,\quad b_{n+1}=(C_{a_1\dots a_{n+1}}^{a})\,.
\ee 
 To prove the assertion we identify the conditions that follow from the vanishing of the square of the vector field with the local form of the defining conditions of a gapped Lie $n$-algebroid according to Definition \ref{gapped def}. We also highlight already that the derivatives of the anchor components and of the binary structure constants appear in the homological vector field. The latter are related to the local coordinate expression for the basic curvature (see Section \ref{sec6} for its precise definition). The part of the proof that refers to the ground Lie algebroid is obvious (it is exactly the same as for any Lie algebroid) and we omit it. According to these, we first find:
 \bea 
\big(Q^{\widehat{\E}}\big)^2\dot{x}^{\m}&=&a^b\dot{a}^c\bigg(-\rho_{a}{}^{\m}C^a_{bc}+\rho_b{}^{\n}\partial_{\n}\rho_c{}^{\m}-\rho_c{}^{\n}\partial_{\n}\rho_b{}^{\m}\bigg) 
\nn\\[4pt] 
&+&  a^{a}a^{b}\dot{x}^{\n}\bigg(\rho_a{}^{\k}\partial_{\k}\partial_{\n}\rho_{b}{}^{\m}-\frac 12 C_{ab}^{c}\partial_{\n}\rho_c{}^{\m}+\partial_{\n}\rho_{a}{}^{\k}\partial_{\k}\rho_b{}^{\m}-\frac 12 \rho_c{}^{\m}\partial_{\n}C_{ab}^{c}\bigg) 
\nn\\[4pt] 
&+&  a^{a_1}\dots a^{a_{n+1}}\bigg(-\frac 1{(n+1)!}\rho_a{}^{\m}C^a_{a_1\dots a_{n+1}}+\frac{(-1)^n}{n!}\rho_{a_{n+1}}{}^{\n}\partial_{\n}\rho_{a_1\dots a_n}{}^{\m} \nn\\[4pt] && \qquad\quad  +\, \frac{(-1)^n}{(n-1)!}C_{a_1 a_{n+1}}^a\rho_{a a_2\dots a_n}{}^{\m}-\frac{(-1)^{n}}{n!}\rho_{a_1\dots a_n}{}^{\n}\partial_{\n}\rho_{a_{n+1}}{}^{\m}\bigg)\,.
 \eea 
 The first and second parentheses vanish due to the ground Lie algebroid relations, the first being the condition that the anchor $\r_2$ is a homomorphism and the second being the derivative of the first. Vanishing of the third parentheses, obviously taking into account the manifest antisymmetrization of indices, is the local coordinate expression of the mixed morphism condition \eqref{mixed morphism} with the replacement of $\Pi$ by $\rho_2$ and of $R$ by $\r_{n+1}$ as before. Finally we find 
 \bea 
\big(Q^{\widehat{\E}}\big)^2\dot{a}^{a}&=&a^{b}a^{c}\dot{a}^{d}\bigg(-\rho_b{}^{\m}\partial_{\m}C^{a}_{cd}-\frac 12 \rho_{d}{}^{\m}\partial_{\m}C^{a}_{bc}-\frac 12 C^{e}_{bc}C^{a}_{de}-C^{e}_{db}C^{a}_{ce}\bigg) \nn\\[4pt] 
&+&a^ba^ca^d\dot{x}^{\n}\bigg(-\frac 12 \rho_b{}^{\m}\partial_{\m}\partial_{\n}C^a_{cd}-\frac 12 C^{e}_{bc}\partial_{\n}C^a_{de}-\frac 12 \partial_{\n}\rho_{b}{}^{\m}\partial_{\m}C^a_{cd}-\frac 12 C^a_{de}\partial_{\n}C^{e}_{bc}\bigg) \nn\\[4pt] 
&+&a^{a_1}\dots a^{a_{n+2}}\bigg(\frac{(-1)^{n+1}}{(n+1)!}\rho_{a_{n+2}}{}^{\m}\partial_{\m}C^a_{a_1\dots a_{n+1}}+\frac{(-1)^{n+1}}{2 n!}C^b_{a_1a_{n+2}}C^{a}_{b a_2\dots a_{n+1}}\nn\\[4pt] && \qquad \qquad -\,\frac{1}{2n!}\rho_{a_1\dots a_{n}}{}^{\m}\partial_{\m}C^a_{a_{n+1}a_{n+2}}+\frac{(-1)^n}{(n+1)!}C^a_{a_{n+2}b}C^b_{a_1\dots a_{n+1}}\bigg)\,.
 \eea 
 The first and second parentheses once again vanish due to the ground Lie algebroid data, the first being the Jacobi identity of the binary Lie bracket and the second the derivative of the first. The third parentheses this time give the local coordinate expression of the mixed Jacobi identity \eqref{mixed jacobi} and complete the agreement with the definition of a gapped Lie $n$-algebroid.  
\begin{rmk}
    The form of the homological vector field is generic for $n>3$. For $n\le 3$, the structure can proliferate but we do not discuss this further in this paper. See \cite{Thanasis R-Poisson} for more on this in the special case of $E=T^{\ast}M$ and $M$ a Poisson manifold.
\end{rmk}
As evident from all discussions up to this point, the case of $E=T^{\ast}M$ with $M$ a Poisson or an R-Poisson manifold is special. Since Poisson is a subcase of R-Poisson we discuss directly the latter and its generalization to the twisted case. Moreover, note that when $E=T^{\ast}M$ there exist another dg manifold with the same local structure as $T[n-1]E[1]$, namely the second order cotangent bundle $T^{\ast}[n]T^{\ast}[1]M$. Indeed, the local coordinate charts of the two cases are indistinguishable. This is a graded geometric manifestation of the fact that there exists a diffeomorphism between $TT^{\ast}M$ and $T^{\ast}T^{\ast}M$ \cite{legendre,legendre2}. 
Then, twisted R-Poisson $n$-algebroids can be described as dg manifolds of degree $n$ arising from a second order cotangent bundle $(T^{\ast}[n]T^{\ast}[1]M,Q^{{RP_n}}_{H})$  with homological vector field 
  \bea 
Q_H^{{RP_n}}&=&Q^{\ast}+\big((-1)^{n}\Pi^{\n\m}z_{\n}-\partial_{\n}\Pi^{\m\k}a_{\k}y^{\n}+\frac 1{n!}R^{\m\n_1\dots \n_{n}}a_{\n_1}\dots a_{\n_{n}}\big)\frac{\partial}{\partial y^{\m}} \nn\\[4pt] 
  &-& \big(\partial_{\m}\Pi^{\n\k}a_{\n}z_{\k}+\frac{(-1)^{n}}{2}\partial_{\m}\partial_{\n}\Pi^{\k\l}y^{\n}a_{\k}a_{\l}-\frac{(-1)^{n}}{(n+1)!}f_{\m}{}^{\n_1\dots\n_{n+1}}a_{\n_1}\dots a_{\n_{n+1}}\big)\frac{\partial}{\partial z_{\m}}\,,\nn\\
  \eea 
  where $(x^{\m},a_{\m},y^{\m},z_{\m})$ are degree $(0,1,n-1,n)$ coordinates on the dg manifold, 
  \be 
  f_{\m}{}^{\n_1\dots\n_{n+1}}=\partial_{\m}R^{\n_1\dots \n_{n+1}}-H_{\m}{}^{\n_1\dots \n_{n+1}}\,,
  \ee 
  and $Q^{\ast}$ is the homological vector field of the underlying Poisson Lie algebroid.
\begin{rmk}\label{Poisson lie n}
    We note that for vanishing $R$ this is still a homological vector field with the sole condition that $\Pi=(\Pi^{\m\n})$ is a Poisson structure. In other words we can realize a Poisson manifold not only as a Lie algebroid but also as a Lie $n$-algebroid. We call this the \textbf{Poisson Lie $n$-algebroid} in what follows and denote its homological vector field as $Q^{\ast}_n:=Q^{{RP_n}}_{H}|_{H=0=R}$. This structure is the backbone of the higher dimensional Poisson sigma models constructed in Section \ref{sec21}.
\end{rmk}

We close this section by emphasizing  the difference between the transition from Poisson to twisted Poisson structure and the transition between R-Poisson and twisted R-Poisson structure. In the former case we remain in the category of Lie algebroids, while on the contrary in the latter case we switch from genuine Lie $n$-algebroids to almost Lie $n$-algebroids. This is not unexpected once we recall the similar and actually related instance of pre-Courant algebroids \cite{Vaisman:2004msa} and specifically of 4-form twisted Courant algebroids \cite{Hansen:2009zd}, which were shown in \cite{BG} to correspond to symplectic almost Lie 2-algebroids, unlike Courant algebroids which are genuine symplectic Lie 2-algebroids \cite{dimaphd}. Twisted R-Poisson $n$-algebroids are similar to this second instance of structural transition. The case $n=2$ in particular has a direct relation to pre-Courant algebroids. Indeed for $n=2$ the gap vanishes and the two brackets are a binary and a ternary one. What we call an R-Poisson $n$-algebroid here can then be seen as a Dirac structure in a Lie 2-algebroid and more generally in a Lie $n$-algebroid, a perspective employed and elaborated in Ref. \cite{Ikeda:2021rir}.

\section{Polytorsion tensor for n-form connections}
\label{sec5} 

Given a Lie algebroid $E$ we can define the notions of an  $E$-connection and an $E$-covariant derivative on some vector bundle $V$ including the special case $V=E$.  We shall often refer to them as ``$E$-on-$V$'' connections for simplicity. Moreover, an $E$-representation is an $E$-connection of $V$ whose curvature vanishes, a flat $E$-on-$V$ connection. When the underlying manifold $M$ is a point, the Lie algebroid $E$ is a Lie algebra $\mf g$ and a flat $E$-connection as above corresponds to the usual notion of Lie algebra representation on some vector space $V$. The case $V=E$ then is the analog of representing the Lie algebra on itself, as is the case with the adjoint representation.  
To be more specific, an $E$-connection on $E$, denoted $\nabla^{\E}: \G(E)\to \G(T^{\ast}M\otimes E)$ is a map that for $f\in C^{\infty}(M)$ satisfies 
\be 
\nabla^{\E}(fe)=f\,\nabla^{\E}e+\dd f\otimes e\,.
\ee 
It gives rise in an obvious way to an $E$-covariant derivative such that 
\bea  
\nabla^{\E}_{fe}e'&=&f\,\nabla^{\E}_{e}e'\,, \\[4pt] 
\nabla^{\E}_{e}(f e')&=& f\,\nabla^{\E}_{e}e'+\rho(e)f \, e'\,.
\eea 
Hence we shall be working with covariant derivatives but also sometimes refer to them as connections in the sense described here. For $E$-connections on $E$ one may define $E$-torsion and $E$-curvature tensors in the usual way. We do not report the formulas yet because they will soon follow as special cases of the generalization we will consider. Note that an $E$-connection on $E$ is not necessarily induced by an ordinary connection $\nabla: \G(TM)\times \G(E)\to \G(E)$ on $E$ with the aid of the anchor map of the Lie algebroid. When it is, we shall denote this canonical induced $E$-on-$E$ connection as $\mathbullet{\nabla}^{\E}$, in particular 
\be 
\mathbullet\nabla^{\E}_e\equiv \nabla_{\rho(e)}\,.
\ee 
In that case we shall be able to relate the $E$-connection to the so-called basic curvature of an ordinary connection on $E$ that measures the compatibility of the connection with the Lie bracket defined on $E$. 

This brief introduction to $E$-connections and $E$-covariant derivatives combined with the discussion in Section \ref{sec3} brings us to the question of whether we can meaningfully and covariantly differentiate a 1-form along an $n$-form. In other words, whether we can define an $E_n$-connection on $E=T^{\ast}M$ such that it has some desired properties like linearity and the Leibniz rule. Equipped with the map $R: E_n\to TM$ of the anchored bundle $(E_n,R)$ we can indeed define such an $E_n$-covariant derivative as 
\begin{align} 
\nabla^{\E_n}: \G(E_n\otimes T^{\ast}M) &\to \G(T^{\ast}M) \nn\\[4pt]
(\hat e, e) &\mapsto \nabla^{\E_n}_{\hat e}e\,,
\end{align} 
such that 
\bea 
\nabla^{\E_n}_{f\hat e}e&=&f\nabla^{\E_n}_{\hat e}e\,, \\[4pt]  
\nabla^{\E_n}_{\hat e}(fe)&=&f\,\nabla^{\E_n}_{\hat e}e + R(\hat e)f\, e\,.
\eea 
  We know of at least one simple example of such an $E_n$-connection; the canonical one induced by an ordinary connection with the assistance of the higher anchor map $R$, which in accord with previous notation reads 
\be 
\mathbullet\nabla^{\E_n}_{\hat e}=\nabla_{R(\hat e)}\,.
\ee 
This is of course not the most general such connection but it serves as a proof of principle and it is the one we will actually need in the present paper. 

For our purposes in this paper, we would like to know whether there exists a meaningful notion of torsion for such connections. In general, when for two vector bundles $V_1$ and $V_2$ with suitable structure one defines a notion of $V_1$-covariant derivative on sections of $V_2$, a torsion tensor is not defined unless $V_1=V_2$. That was after all the reason we focused initially on $E$-connections on $E$ rather than on any other vector bundle. But now the situation has changed; we have an $E_n=\w^{n}T^{\ast}M$-connection on $E_1=E=T^{\ast}M$ and $E_n\ne E_1$ for $n\ne 1$. Nevertheless, the fact that $E_n$ is a tensor power of $E$ opens up a new possibility described in the following definition: 
\begin{defn}
 The $E_n$-polytorsion of an $E_n$-connection $\nabla^{\E_n}$ on the cotangent bundle of a smooth manifold $M$, where $E_n=\w^nT^{\ast}M$, is a map ${T^{\nabla^{\E_n}}}:\G(\otimes^{n+1}T^{\ast}M)\to \G(T^{\ast}M)$ given as
 \be
 {T^{\nabla^{\E_n}}}(e^{(1)},\dots,e^{(n+1)})=\sum_{r=1}^{n+1}(-1)^{n-r+1}\nabla^{\E_n}_{\hat e[r]}e^{(r)}-b_{n+1}(e^{(1)},\dots , e^{(n+1)})\,,
 \ee 
 where $e^{(r)}\in\G(T^{\ast}M)$ are 1-forms, $\hat e[r]\in \G(\w^{n}T^{\ast}M)$ is the decomposable $n$-form given by the exclusion of the $r$-th entry in the wedge string of $e^{(r)}$s
 and the bracket on the right hand side is the $(n+1)$-ary Koszul bracket. 
 \end{defn}
  Note that the definition makes sense also in presence of twists for the $(n+1)$-ary bracket as described in Section \ref{sec3}. We shall not use any special terminology or notation then; the polytorsion should always be computed with the right bracket available depending on the structure at hand.
 Note, moreover, that for $n=1$, in which case $E_n=T^{\ast}M$, the $E_1$-polytorsion is identical to the $E$-torsion as usually defined: 
 \be 
 T^{\nabla^{\E}}(e,e')= \nabla^{\E}_{e}e'- \nabla^{\E}_{e'} e - b_2(e,e')\,.
 \ee  
As expected, the $E$-torsion does not necessarily vanish when the ordinary torsion $T$ of an ordinary connection on $M$ vanishes; this remains true for general values of $n$.   
 To appreciate the structure of the $E_n$-polytorsion, we write down its expanded form in case $n=2$ in the case of torsionless connection $\nabla$ on $M$---i.e. when the twist in the bracket vanishes. Then the $E_2$-polytorsion is given as 
 \be 
 {T^{\nabla^{\E_2}}}(e^{(1)},e^{(2)},e^{(3)})=\nabla^{\E_2}_{e^{(2)}\w e^{(3)}}e^{(1)}+ \nabla^{\E_2}_{e^{(3)}\w e^{(1)}}e^{(2)}+ \nabla^{\E_2}_{e^{(1)}\w e^{(2)}}e^{(3)}-b_3(e^{(1)}, e^{(2)} , e^{(3)})\,,
 \ee 
 where the ternary Koszul bracket appears in \eqref{higherKoszul3}. 

 The $E_n$-polytorsion is a tensorial object that satisfies the correct $C^{\infty}$-linearity properties in all its entries, specifically 
 \be 
T^{\nabla^{\E_n}}(e^{(1)},\dots,f e^{(r)},\dots,e^{(n+1)})=f\, T^{\nabla^{\E_n}}(e^{(1)},\dots,e^{(n+1)})\,.
 \ee  
 This is easily proven as follows: 
\begin{align*}
    T^{\nabla^{\E_{n}}}(e^{(1)}, \dots, fe^{(r)},\dots, e^{(n+1)}) & 
    = \sum_{\substack{j=1 \\ j\neq r}}^{n+1}(-1)^{n-j+1}\nabla^{\E_{n}}_{{\hat{e}_f[j]}}e^{(j)}+(-1)^{n-r+1}\nabla^{\E_{n}}_{\hat{e}[r]}(fe^{(r)})\\[4pt] 
    &- fb_{n+1}(e^{(1)}, \dots, e^{(r)}, \dots, e^{(n+1)})-(-1)^{n+r-1}R(\hat{e}[r])fe^{(r)} \\[4pt] 
    & = \sum_{\substack{j=1 \\ j\neq r}}^{n+1}(-1)^{n-j+1}{f}\,\nabla^{\E_{n}}_{\hat{e}[j]}e^{(j)}+(-1)^{n-r+1}f\,\nabla^{\E_{n}}_{\hat{e}[r]}e^{(r)}\\[4pt] 
    & +\cancel{(-1)^{n-r+1}R(\hat{e}[r])fe^{(r)}}-f\,b_{n+1}(e^{(1)}, ..., e^{(n+1)}) \nn\\[4pt] & -\cancel{(-1)^{n-r-1}R(\hat{e}[r])fe^{(r)}}\\ 
    & = f\,\sum_{j=1}^{n+1}\nabla_{\hat{e}[j]}^{\E_{n}}e^{(j)}-f\,b_{n+1}(e^{(1)}, \dots, e^{(n+1)})\\
    & = f\,T^{\nabla^{E_{n}}}(e^{(1)}, \dots, e^{(n+1)})\,,
\end{align*}
using the properties of the covariant derivative and of the higher bracket. 
    Note that in the first line $\hat{e}_{f}[j]$ with $j\ne r$ is rather obviously understood to be the decomposable $n$-form that lacks the entry $e^{(j)}$ and whose $r$-th entry is $f\, e^{(r)}$.

 Let us now compute the $E_n$-polytorsion for arbitrary $n$ using explicit local coordinate expressions for the anchor and the $E$-covariant derivative. In this part we work with the induced $E_n$-connection $\mathbullet\nabla^{\E_n}$; the general case proceeds in similar steps. First we expand a general $n$-form as $\hat e=\frac 1{n!}\hat e_{\m_1\dots\m_n}\dd x^{\m_1}\w\dots\w\dd x^{\m_n}$. Our choice of convention in the anchored bundle $E_n$ is such that the anchor $R$ acts on this $n$-form as 
 \be 
R(\hat e)=\frac 1{n!} R^{\m_1\dots \m_{n+1}}\hat e_{\m_1\dots\m_n}\partial_{\m_{n+1}}\,.
 \ee 
 Then the $E_n$-covariant derivative on $E$ has components 
 \be 
\mathbullet\nabla^{\E_n}_{\hat e}e=\nabla_{R(\hat e)}e=\frac 1{n!} R^{\m_1\dots \m_{n+1}}\hat e_{\m_1\dots \m_n} (\partial_{\m_{n+1}}e_\n-\Gamma_{\m_{n+1}\n}^{\rho}e_{\rho})\dd x^{\n}\,,
 \ee 
 where $\G_{\m\n}^{\rho}$ are the coefficients of the ordinary connection $\nabla$, which could also have torsion $T$. Of course, if this torsion $T$ is nonvanishing then the Koszul bracket is twisted accordingly. 
 In any case, a straightforward computation leads to the general result 
\be 
T^{\nabla^{\E_n}}=-\mathring{\nabla}R\,,
\ee 
where $\mathring{\nabla}$ is the torsionless part of the affine connection $\nabla$ on $M$. In other words, writing the connection coefficients of $\nabla$ as
\be 
\G^{\m}_{\n\rho}=\mathring{\Gamma}^{\m}_{\n\rho}+\frac 12 T^{\m}_{\n\rho}\,,
\ee 
with $\mathring{\Gamma}^{\m}_{[\n\rho]}=0$, then $\mathring{\Gamma}$ are the connection coefficients of $\mathring{\nabla}$.  In local coordinates, this reads 
\be 
T^{\nabla^{\E_n}}(e^{(1)},\dots, e^{(n+1)})=-\frac 1{n!}\mathring{\nabla}_{\rho}R^{\m_1\dots\m_{n+1}}e^{(1)}_{\m_1}\dots e^{(n+1)}_{\m_{n+1}}\dd x^{\rho}\,.
\ee 
For example, for $n=1$ and a (twisted or untwisted) Poisson structure given by $\Pi$, the $E$-torsion is equal to $-\mathring{\nabla}\Pi$, as found in \cite{Ikeda:2019czt}.
It thus becomes clear that the $E$-torsional data of the Poisson Lie algebroid $(T^{\ast}M, b_2,\Pi)$ follow as a corollary from the definition of $E_n$-torsion for $n=1$.  We emphasize that this holds both in the untwisted and twisted cases; the torsion of the base manifold drops out from the $E$-torsion, which only depends on the anchoring data. This is in agreement with the findings of \cite{Ikeda:2019czt}.

For $n>1$ there is a final twist to the above story, which is important for the field-theoretical considerations of Section \ref{sec2}. As we showed above, the twist $Tw_{n+1}$ of the $(n+1)$-ary Koszul bracket cancelled precisely the torsional part of $\nabla$ in the definition of the $E_n$-polytorsion. However, in the general case when there is a closed $(n+2)$-form $H$ available we saw in Section \ref{sec3} that there exists an additional twist given by $H_{n+1}$ in \eqref{Hn+1}. The two twists are identical for $n=1$, and therefore one of them is redundant then, but totally independent for $n>1$. The net result is that in the context of homotopy Poisson structures such as the twisted R-Poisson one, the $E_n$-polytorsion has additional terms. Specifically for twisted R-Poisson this is 
\be  
T^{\nabla^{\E_n}}=-\mathring\nabla R-\langle\Pi^{\otimes (n+1)},H\rangle\,.
\ee 
This is precisely the combination that arose in the context of higher-dimensional Hamiltonian mechanics on a target space with twisted R-Poisson structure that we recalled in section \ref{sec2} and thus it provides the sought after geometrical interpretation of the corresponding interaction coefficients in the BRST formalism. 

\section{Basic connections \& basic polycurvature tensor}
\label{sec6}

Going one step further, we would now like to define a higher analog of the basic curvature tensor for arbitrary $n$. The basic curvature tensor for a Lie algebroid measures the departure from compatibility between an ordinary connection on the Lie algebroid and the binary Lie bracket on it \cite{Blaom,Crainic}. In other words it is the answer to the question: \emph{Is a given $TM$-on-$E$ connection on a Lie algebroid $E$ a derivation of its Lie bracket?} For Courant algebroids the corresponding notion is found in \cite{Chatzistavrakidis:2023otk}. Moreover, an approach to this problem for Lie $n$-algebroids was suggested in \cite{Jotz}. 

Within the total space of a smooth fiber bundle fundamental role have two special kinds of forms, namely vertical forms (whose pushfoward to the base vanishes) and basic differential forms, that is pullbacks of the forms on the base. For principal bundles, basic differential forms make a subcomplex of the de Rham complex and Chern-Weil map in theory of characteristic classes takes values in its cohomology (basic cohomology). Secondary characteristic classes of Lie algebroids are studied in \cite{CF} where a related notion of basic connection is introduced. Two canonical (classes of) examples are introduced there starting from an ordinary connection $\nabla$ on a Lie algebroid $(E,b_2,\rho_2)$, one being an $E$-connection on $E$ and the other being an $E$-connection on the tangent bundle $TM$. They are defined respectively as 
\bea  
\label{basic E on E} \overline\nabla^{\E}_{e}e'&=&\mathbullet\nabla^{\E}_{e'}e+b_2(e,e') \\[4pt]
&=& \mathbullet\nabla^{\E}_e e'-T^{\mathbullet\nabla^{\E}}(e,e')\,, 
\\[4pt] 
\label{basic E on TM}
\overline\nabla^{\E}_{e}X&=&\rho_2(\nabla_{X}e)+[\rho_2(e),X]\,,
\eea 
where $X\in\G(TM)$ is a vector field on $M$ and we use the notation $\rho_2$ for the anchor of the Lie algebroid. Note that we expressed the first basic connection in two alternative ways, one with the Lie bracket and one with the $E$-torsion, since they are both useful in different places below. 
This basic connection is tightly related to the notion of the opposite (or conjugated) connection, one which when averaged with the original connection the result has vanishing torsion. Specifically, note that if we define the average connection 
\be 
\nabla^{\E,\text{avg}}=\frac 12 \left(\mathbullet\nabla^{\E}+\overline\nabla^{\E}\right)\,,
\ee 
then it is simple to see that 
\be 
T^{\nabla^{\E,\text{avg}}}=0\,.
\ee 
We will use this for guidance in the more general case. Regarding the second part of the basic connection, the one given in \eqref{basic E on TM}, it can be uniquely determined as follows. Asking whether the connection $\nabla$ on $E$ is a derivation of the Lie bracket $b_2$ there is a tensor that controls the answer to this question, called the basic curvature tensor \cite{Crainic}. $C^{\infty}(M)$-linearity singles out the use of the basic connection \eqref{basic E on TM}. We will see this explicitly right below in an even more general setting, therefore we skip the details at this point. 

Consider now that $E=T^{\ast}M$ for a twisted R-Poisson manifold $(M,\Pi,R,H)$. The basic connections above are sufficient for the ground Lie algebroid but not for the twisted R-Poisson $n$-algebroid, which comprises a higher Koszul bracket. To account for this, we consider an ordinary $TM$-on-$E$ connection $\nabla$ and extend it in the obvious way to the exterior power bundle $E_n=\w^n T^{\ast}M$. Then aside from the basic $E$-on-$E$ and $E$-on-$TM$ connections introduced above we also have the following $E_{n}$-on-$E$ and $E_n$-on-$TM$ basic connections:
\bea \label{EnonE}
\overline\nabla^{\E_n}_{\hat e}e&=&\mathbullet\nabla^{\E_n}_{\hat e}e-\k_n T^{\mathbullet\nabla^{\E_{n}}}(e^{(1)},\dots,e^{(n)},e)\,,\\[4pt]
\overline{\nabla}^{\E_n}_{\hat{e}}X&=&R(\nabla_{X}\hat{e})+[R(\hat{e}),X]\,,
 \label{EnonTM}
\eea 
respectively, where $\hat{e}=e^{(1)}\w\dots\w e^{(n)}$ is decomposable and $\k_n$ is a parameter such that $\k_1=1$ for consistency with the $E$-on-$E$ basic connection.{\footnote{Note that like ordinary connections, these basic connections on $E$  may be extended in an obvious way to the exterior bundle $E_n$. This means that we can also define  $E$-on-$E_n$ and $E_n$-on-$E_n$ basic connections, with their action on decomposable $n$-forms being
\bea 
\overline{\nabla}^{\E}_{e}\hat{e}&=&\sum_{r=1}^{n}(-1)^{r-1}\overline{\nabla}^{\E}_ee^{(r)}\w \hat{e}[r]\,,\\[4pt]
\overline\nabla^{\E_n}_{\hat{e}}\hat{e}'&=&\sum_{r=1}^{n}(-1)^{r-1}\overline\nabla^{\E_n}_{\hat{e}}e'^{(r)}\w\hat{e}'[r]\,.
\eea 
Obviously a similar extension holds for other connections such as the ordinary one $\nabla$.}} Essentially the upper formula gives a 1-parameter family of connections, even though we do not introduce a special notation for this. It may be written in the alternative form 
\be 
\overline\nabla^{\E_{n}}_{\hat e}e=\k_n\sum_{r=1}^{n}(-1)^{n-r}\mathbullet\nabla^{\E_n}_{{\hat e}[r]\w e}e^{(r)}+(1-\k_n)\mathbullet\nabla_{\hat{e}}e+\k_n b_{n+1}(e^{(1)},\dots, e^{(n)},e)\,,
\ee 
which is closer in spirit to the starting definition of the opposite connection.
If for example we demand that the higher average connection 
\be 
\nabla^{\E_n,\text{avg}}=\frac 12 \left(\mathbullet\nabla^{\E_n}+\overline\nabla^{\E_n}\right)\,,
\ee 
has vanishing polytorsion tensor,
\be 
T^{\nabla^{\E_n,\text{avg}}}=0\,,
\ee
then we find that $\k_n=2/(n+1)$. 
Equipped with the basic connections above, particularly  the $E_n$-on-$TM$ one,  we can give the following definition: 
\begin{defn}
The basic $n$-polycurvature of a connection $\nabla$ on $T^{\ast}M$ equipped with the $(n+1)$-ary Koszul bracket is a map $ {S^{\nabla}}:\G(\otimes^{n+1}T^{\ast}M\otimes TM)\to \G(T^{\ast}M)$ 
given by the formula
 \bea
 {S^{\nabla}}(e^{(1)},\dots,e^{(n+1)})X&=&\nabla_{X}b_{n+1}(e^{(1)},\dots,e^{(n+1)})  -\sum_{r=1}^{n+1}(-1)^{n-r}\nabla_{\overline{\nabla}^{\E_n}_{\hat{e}[r]}X}e^{(r)} -  \nn\\[4pt] &&- \sum_{r=1}^{n+1} 
 b_{n+1}(e^{(1)},\dots,e^{(r-1)},\nabla_{X}e^{(r)},e^{(r+1)},\dots,e^{(n+1)})\,,\,\,\,\label{Sep1}
 \eea 
 where $X\in \G(TM)$ is a vector field on $M$, $e^{(r)}\in\G(T^{\ast}M)$ are 1-forms, $\hat e[r]\in \G(\w^{n}T^{\ast}M)$ is the decomposable $n$-form \eqref{deco}, the $E_n$-connection on $TM$ is given in \eqref{EnonTM}
 and $b_{n+1}$ is the $(n+1)$-ary (possibly twisted) Koszul bracket.
\end{defn}

In fact, the form of the basic $E_n$-on-$TM$ connection in Eq. \eqref{EnonTM} is completely fixed by the requirement of linearity of the basic $n$-polycurvature, including of course the standard case of $n=1$ as alluded to earlier in the present section. As we did with the $E_n$-polytorsion, we can compute the basic $n$-polycurvature and bring it into the following form in local coordinates 
\bea 
&& S^{\nabla}(e^{(1)},\dots, e^{(n+1)})X=\frac 1{n!}X^{\m}e^{(1)}_{\m_1}\dots e^{(n+1)}_{\m_{n+1}}\times \nn\\[4pt] && \hspace{110pt}\times \left(\nabla_{\m}\mathring{\nabla}_{\rho}R^{\m_1\dots \m_{n+1}}-R^{\n\m_2\dots \m_{n+1}}{\cal R}^{\m_1}_{\rho\m\n}-\dots - R^{\m_1\dots \m_n\n}{\cal R}^{\m_{n+1}}_{\rho\m\n}\right)\dd x^{\rho},\nn
\eea 
where ${\cal R}$ is the ordinary curvature of the ordinary affine connection $\nabla$ on $M$, with or without torsion. 
Combining this with the result we found for the $E_n$-polytorsion and expressing everything in coordinate-free form, we conclude that 
\be \label{Sep2}
S^{\nabla}=-\nabla T^{\mathbullet\nabla^{\E_n}}-(n+1)\,\text{Alt}\langle R,{\cal R}\rangle\,,
\ee 
where $\text{Alt}$ is antisymmetrization with weight 1 in $\otimes^{n+1}T^{\ast}M$. It is also obvious that the basic polycurvature is a linear object.
Importantly, for the special case of $n=1$ and Poisson or twisted Poisson bivector $\Pi$, the general formulas \eqref{Sep1} or \eqref{Sep2} reproduce (in one case up to an overall sign that depends solely on different conventions) the result for the ordinary basic curvature found in the literature, see \cite{Blaom, Crainic} and \cite{Kotov:2016lpx} respectively for each formula. We have generalized this to the higher case here and we comment further in the concluding section about the possible relation to the concept of representations up to homotopy. The main conclusion here is that this basic polycurvature tensor is precisely the coefficient \eqref{basic poly in brst} that we encountered in the BRST formalism of the twisted R-Poisson sigma model.

\section{Conclusions and outlook} 
\label{sec7} 

In this paper we demonstrated that the coefficients of the interaction terms in the BV/BRST formulation of a class of topological sigma models in any dimension acquire a geometric interpretation as tensors for generalized connections on gapped Lie $n$-algebroids. The gauge field theories we considered include higher-dimensional Poisson sigma models without or with Wess-Zumino term and their deformations by generalized R-flux multivectors. These theories are attractive because they constitute a vertical slice in the space of AKSZ models in any dimension, they are simple and tractable and they feature all unconventional properties that go beyond Yang-Mills-like theories, namely an open gauge algebra, reducibility of gauge generators and dependence of coefficients on the fields. What is more, the openness of the gauge algebra is nonlinear in the field equations and this is arguably the simplest class of gauge theories with this characteristic. Moreover, although in the untwisted case they correspond to QP manifolds and they are vanilla AKSZ theories, generalized fluxes obstruct this via twists and the target is just a Q manifold, in which case the AKSZ construction is not applicable at face value. 

To determine the geometric meaning of the several field-dependent coefficients in the interaction terms of the master action of the models, we introduced the notion of a gapped (almost) Lie $n$-algebroid. In the case of interest in this paper, this structure features two anchor maps, one being a usual map from the cotangent bundle of the twisted R-Poisson manifold to its tangent bundle given by a Poisson 2-vector $\Pi$ and the other being a higher anchor from the $n$-th exterior power bundle of $T^{\ast}M$ to $TM$, given by an $(n+1)$-vector $R$. It moreover contains two brackets, one binary Koszul bracket and one $(n+2)$-form-twisted $(n+1)$-ary higher Koszul bracket of 1-forms. Both brackets satisfy a Leibniz rule with respect to the two anchors respectively and they verify the conditions that guarantee that the space of sections has the structure of an $L_{\infty}$ algebra in the untwisted case. In the twisted case, an analog of an almost Lie algebroid is encountered. A gapped Lie $n$-algebroid is then an example of an $L_{\infty}$ algebroid in the sense of \cite{Voronov2} and a gapped almost Lie $n$-algebroid a twisted version thereof that violates the highest Jacobi identity with the violation not being necessarily controlled by a higher bracket. This is a different construction than the one that provides a correspondence between homological vector fields and $L_{\infty}$ algebroids with all brackets but the binary one being $C^{\infty}$-linear, as e.g. in \cite{BP}.

We studied connections on the structure described above. Starting with an ordinary affine connection on 
the cotangent bundle, we showed that the two anchors induce canonically two generalized connections, one $E$-on-$E$ and one $E_{n}$-on-$E$ connection, where $E=T^{\ast}M$ and $E_{n}=\w^nT^{\ast}M$. For the first connection there is a natural $E$-torsion tensor defined in the usual way. We showed that for the second connection there also exists a tensor, which involves the higher Koszul bracket and we call it the $E_{n}$-polytorsion because it reduces to the $E$-torsion for $n=1$. 

Apart from the canonically induced connections, the affine connection on $T^{\ast}M$ also induces two sets of two ``basic'' connections. The first set is based on the usual anchor and it was originally considered in \cite{CF}. They are an $E$-on-$E$ and an $E$-on-$TM$ connection, with the latter crucially participating in defining the tensor that measures whether the affine connection is a derivation of the Lie bracket for Lie algebroids, known as the basic curvature tensor \cite{Crainic}. Here we defined a second set of an $E_{n}$-on-$E$ and an $E_{n}$-on-$TM$ connection and introduced the basic polycurvature tensor that measures whether the affine connection is a derivation of the higher $(n+1)$-ary Koszul bracket. All torsion and basic curvature tensors described above appear in the BRST formalism of the twisted R-Poisson sigma models, as summarized in Table \ref{Table1}. 

\begin{table}
	\begin{center}	\begin{tabular}{| c | c | c |}
			\hline 
			\multirow{3}{10.5em}{\centering{Interaction (examples)}} & \multirow{3}{10.5em}{\centering{Coefficient}}&\multirow{3}{10.5em}{\centering{Geometric meaning}} \\  &  &    \\ && \\ \hline
			\multirow{3}{5.5em}{} &  &   \\ $ZA, X^{\dagger}\e, Y^{\dagger}\psi$ & $\Pi^{\m\n}$ & Anchor   \\ && \\\hline 
			\multirow{3}{5.5em}{} && \\ $A^{n+1}, \chi^{\dagger}\e^2, Y^{\dagger}A\e$ & $R^{\m_1\dots\m_{n+1}}$ &  Higher anchor   \\ &&
			\\\hline 
   \multirow{3}{10.5em}{} & \multirow{3}{10.5em}{}&\multirow{3}{10.5em}{} \\ $YA^2, \e^{\dagger}\e^2, \chi^{\dagger}\e\chi, Y^{\dagger}Y\e, Y^{\dagger}A\chi$ & $\nabla_{\r}\Pi^{\m\n}$ & $E$-torsion   \\ && \\ \hline
   \multirow{3}{10.5em}{} & \multirow{3}{10.5em}{}&\multirow{3}{10.5em}{} \\ $Z^{\dagger}A^2\e^{n-1}, \psi^{\dagger}A\e^n, \psi_{(1)}^{\dagger}\e^{n+1}$ & $\nabla_{\n}R^{\m_1\dots\m_{n+1}}-H_{\n}{}^{\m_1\dots\m_{n+1}}$ & $E_n$-polytorsion   \\ && \\ \hline
   \multirow{3}{10.5em}{} & \multirow{3}{10.5em}{}&\multirow{3}{10.5em}{} \\ $Z^{\dagger}A^{\dagger}\e^2, Z^{\dagger}Y^{\dagger}\e\chi, Z^{\dagger}Z^{\dagger}\e\psi$ & $\nabla_{\k}\nabla_{\l}\Pi^{\m\n}-2\Pi^{\r[\m}{\mc R}^{\n]}{}_{\l\r\k}$ & Basic curvature   \\ && \\ \hline
   \multirow{3}{10.5em}{} & \multirow{3}{10.5em}{}&\multirow{5}{10.5em}{} \\  & $\nabla_{\m}\big(\nabla_{\n}R^{\m_1\dots\m_{n+1}}-H_{\n}{}^{\m_1\dots\m_{n+1}}\big)$   &    \\ $Z^{\dagger}\psi^{\dagger}\e^{n+1}, Z^{\dagger}Z^{\dagger}A\e^n$&& Basic polycurvature \\ &$-(n+1)R^{\r[\m_1\dots\m_n}{\mc R}^{\m_{n+1}]}{}_{\n\r\m}$& \\ && \\  \hline\end{tabular}\end{center}\caption{Summary of the main result on the geometric meaning of the interaction coefficients for twisted R-Poisson sigma models. In the first column we identify some representatives of the several interactions between fields, ghosts and antifields grouping them by the coefficient they come with, which appears in the second column in manifestly covariant form in terms of the affine connection $\nabla$ on $T^{\ast}M$. In the third column we report the geometric interpretation of each coefficient in terms of data on the corresponding gapped almost Lie $n$-algebroid and the tensors of the generalized connections induced by $\nabla$, as described in the main text. Further higher interactions correspond to derivatives of these main geometrical data.}\label{Table1}\end{table}

Interestingly, the basic curvature tensor is essential in defining the adjoint representation in the case of Lie algebroids \cite{Crainic}. This is a representation up to homotopy (ruth), defined as a connection on a chain complex instead of a vector bundle, or as a graded vector bundle with some additional data. In the case of the adjoint representation the chain complex is simply $E\overset{\rho}\rightarrow TM$ with the anchor taken as a degree 1 map and the additional data are the two basic connections, thought of as a single $E$-connection on the chain complex and the basic curvature tensor. Note that a ruth is not associated to a strictly flat connection; the curvature of the connection is controlled precisely by the map $\rho$ and the basic curvature of the ordinary connection from which the basic connections are induced. The presence of the basic curvature and its higher analogons in the BRST differential points to a relation between ruths and homological perturbation theory, essentially already alluded to in \cite{Crainic}. It would be interesting to make this precise in explicit examples of gauge theories, such as the ones treated here. This would also give natural examples of ruths for Lie $n$-algebroids that were only recently defined \cite{Jotz, ruth n}. Although our initial intention was to report on this in the present paper, the topic is broader than the particular cases discussed here and deserves separate treatment. We plan to report on the specific relation between ruths and gauge theory, including Yang-Mills-like models and higher gauge theory, in future work.      

\paragraph{Acknowledgements.} We gratefully acknowledge helpful discussions with Thomas Basile, Panos Batakidis, Noriaki Ikeda, Lara Jonke, Sylvain Lavau, Fani Petalidou, Dima Roytenberg and Zhenya Skvortsov. This work was supported by the Croatian Science Foundation project IP-2019-04-4168
“Symmetries for Quantum Gravity”. This article is also based upon work from the COST Action
21109 CaLISTA, supported by COST (European Cooperation in
Science and Technology).

\end{document}